\documentclass[11pt]{amsart}
\usepackage{amsthm,amssymb}

\usepackage{enumerate}
\usepackage{xspace}
\usepackage[all]{xy}
\usepackage{tikz}
\usetikzlibrary{arrows,shapes,trees,snakes}

\usepackage{pgfpages}

\newif\ifdraftpaper
\draftpaperfalse

\ifdraftpaper
\newcommand{\nb}[1]{\textcolor{blue}{\bf\large \#}\footnote{\textcolor{blue}{#1}}}
\newcommand{\nw}[1]{\textcolor{magenta}{\bf\large \#}\footnote{\textcolor{magenta}{#1}}}

\else
\newcommand{\nb}[1]{}
\newcommand{\nw}[1]{}
\fi

 \title[Two-variable fragments with one transitive relation]{On the satisfiability problem for fragments of two-variable logic with one transitive relation}
 \author[Wies\l aw~Szwast]
 {Wies\l aw~Szwast$^*$}\thanks{$^*$Corresponding author. Institute  of Computer Science,  University of Opole, Oleska 48, 45-052 Opole, Poland; Email {\tt szwast@uni.opole.pl}, Tel {\tt +48 774527210}.}
 \author[Lidia~Tendera]{Lidia~Tendera}
  \address{Institute  of Computer Science, University of Opole, Poland}
 \email{[szwast,tendera]@uni.opole.pl}

\newcommand{\+}{+}
\newcommand{\m}{-}

\newcommand{\mna}{\mbox{$<_{\fA}$}}
\newcommand{\nmna}{\mbox{$\not<_{\fA}$}}
\newcommand{\mnb}{\mbox{$<_{\exb}$}}

\newcommand{\mnaa}{\mbox{$<_{\fA_1}$}}
\newcommand{\mn}{\mbox{$<$}}
\newcommand{\wi}{\mbox{$>$}}
\newcommand{\simw}{\sim}
\newcommand{\sima}{\mbox{$\sim_{\fA}$}}
\newcommand{\simb}{\mbox{$\sim_{\exb}$}}
\newcommand{\lsima}{\mbox{$\lesssim_{\fA}$}}
\newcommand{\lsimb}{\mbox{$\lesssim_{\exb}$}}

\newcommand{\N}{{\mathbb N}}
\newcommand{\Z}{{\mathbb Z}}
\newcommand{\Ur}{\overrightarrow{U}}
\newcommand{\Ul}{\overleftarrow{U}}
\newcommand{\Ulr}{U}
\newcommand{\Un}{\overline{U}}
\newcommand{\phir}{\phi^{\rightarrow}}
\newcommand{\restr}{\!\!\restriction\!\!}
\newcommand{\AAA}{\mbox{\large \boldmath $\alpha$}}
\newcommand{\BBB}{\mbox{\large \boldmath $\beta$}}
\newcommand{\BBBr}{\mbox{\large \boldmath $\beta$}^{\rightarrow}}
\newcommand{\BBBl}{\mbox{\large \boldmath $\beta$}^{\leftarrow}}
\newcommand{\BBBlr}{\mbox{\large \boldmath $\beta$}^{\leftrightarrow}}
\newcommand{\BBBn}{\mbox{\large \boldmath $\beta$}^{-}}
\newcommand{\fA}{\mathfrak{A}}%
\newcommand{\fB}{\mathfrak{B}}%
\newcommand{\fC}{\mathfrak{C}}%
\newcommand{\fD}{\mathfrak{D}}%
\newcommand{\gfdtg}{\ensuremath{\mbox{GF\/}^2\mbox{+TG}}}
\newcommand{\gftg}{\ensuremath{\mbox{GF+TG}}}
\newcommand{\gf}{\ensuremath{\mbox{GF}}}
\newcommand{\gfd}{\ensuremath{\mbox{GF\/}^2}}
\newcommand{\FOt}{\mbox{$\mbox{\rm FO}^2$}}
\newcommand{\GFt}{\mbox{$\mbox{\rm GF}^2$}}
\newcommand{\fod}{\ensuremath{\mbox{FO\/}^2}}
\newcommand{\fodt}{\ensuremath{\mbox{FO\/}^2_{T}}}

\newcommand{\NExpTime}{\textsc{NExpTime}}
\newcommand{\TwoExpTime}{2\textsc{-ExpTime}}
\newcommand{\TwoNExpTime}{2\textsc{-NExpTime}}

\newcommand{\fodeg}{\ensuremath{\mbox{FO\/}^2_{T:tw}}\xspace}
\newcommand{\fodeng}{\ensuremath{\mbox{FO\/}^2_{T:fw}}\xspace}

\renewcommand{\phi}{\varphi}
\newcommand{\obciety}{\!\!\upharpoonright\!\!}

\newtheorem{tw}{Theorem}[section]
\newtheorem{corollary}[tw]{Corollary}
\newtheorem{claim}[tw]{Claim}
\newtheorem{lemma}[tw]{Lemma}
\newtheorem{proposition}[tw]{Proposition}
\newtheorem{definition}[tw]{Definition}

\newcommand{\spl}{{{sp}} }
\newcommand{\In}{{{In}} }
\newcommand{\Out}{{{Out}} }
\newcommand{\SP}{{{Sp}}}

\newcommand{\exb}{{\fA_{+B}}}
\newcommand{\exba}{{\fA_{+B_1}}}

\newcommand{\id}{{{id}} }
\newcommand{\cont}[2]{\langle #1\rangle_{#2}}

\begin{document}

\begin{abstract}
	
We study the satisfiability problem for two-variable
first-order logic over structures with one transitive relation. 
We show that the problem is decidable in \TwoNExpTime{} for the fragment consisting of formulas where existential quantifiers are guarded by transitive atoms. As this fragment enjoys neither the finite model property nor the tree model
property, to show decidability we introduce a novel model construction technique based on the infinite Ramsey theorem. 

We also point out why the technique is not sufficient to obtain decidability for the full two-variable logic with one transitive relation, hence contrary to our previous claim, [{FO$^2$ with one transitive relation is decidable}, STACS 2013: 317-328], 
the status of the latter problem remains open.   

\end{abstract}

\keywords{
two-variable first-order logic, decidability, satisfiability problem, transitivity, computational complexity}

\subjclass[2000]{03B25, 03B70}

\maketitle

\section{Introduction}

\label{sec:intro}

The two-variable fragment of first-order logic, \fod, is the restriction of classical first-order logic over
relational signatures to formulas with at most two variables.
It is well-known that   \fod\ enjoys the finite model
property~\cite{Mor75}, and its satisfiability (hence also finite
satisfiability) problem is \NExpTime-complete~\cite{GKV97}.

One  drawback of \fod\ is that it can neither express
transitivity of a binary relation nor say that a binary relation
is a partial (or linear) order, or an equivalence relation. These
natural properties are important for practical applications, thus attempts have been made 
to investigate \fod\ over restricted classes
of structures in which some distinguished binary symbols are
required to be interpreted as transitive relations, orders,
equivalences, etc. The idea to restrict the class of structures  comes from modal correspondence
theory, where various conditions on the accessibility relations
allow one to restrict the class of  Kripke structures considered, e.g.
to transitive structures for the modal logic K4 or equivalence
structures for the modal logic  S5. Orderings, on the other hand,
are very natural when considering temporal logics, where they
model time flow, but they also are used in different scenarios,
e.g.~in databases or description logics. 

However, the picture 
for \fod\ is more complex. In particular, both the satisfiability and the finite
satisfiability problems for \fod\ are undecidable in the presence
of several equivalence relations, several transitive relations, or several linear orders
\cite{GO98,GradelOR99,Otto01}. These results were later strengthened:
\fod\ is undecidable in the presence of two transitive relations
\cite{Kie05,Kaz06}, three equivalence relations \cite{KO05,KO12}, one
transitive and one equivalence relation \cite{KT09}, 
or three
linear orders \cite{Kie2011}.

On the positive side it is known that \fod\ with one or two
equivalence relations is decidable \cite{KO12,KT09,KMP-HT14}. The
same holds for \fod\ with one linear order \cite{Otto01}. The
intriguing question left open by this research was the case of
\fod\ with one transitive relation, and \fod{} with two linear orders. 

The above-mentioned and additional related results are summarized in Figure \ref{tabela}. 
There, \gfd\ is the two-variable
restriction of the {\em guarded fragment}  \gf~\cite{ABN98}, where all
quantifiers are guarded by atoms, and \gftg\ is the restriction of
\gfd\ with transitive relations, where the transitive relation
symbols are allowed to appear only in guards. As shown in
\cite{ST01, ST04}  undecidability of \fod\ with
transitivity transfers to \gfd\ with transitivity; however, \gftg\
is decidable for {\em any} number of transitive symbols.
Moreover, as noted in \cite{Kie05}, the decision procedure
developed for \gfdtg\ can be applied to \gfd\ with one transitive
relation that is allowed to appear also outside guards, giving
\TwoExpTime-completeness of the latter fragment. 

We denote by \fodt\ the set of \fod-formulas over any signature containing a distinguished binary
predicate $T$ which is always interpreted as a
\emph{transitive} relation. We distinguish two fragments of \fodt\ depending on how existential quantifiers are used: \fodt{} {\em with transitive witnesses}, where existential quantifiers are guarded by transitive atoms (when written in negation normal form), and \fodt{} {\em with free witnesses}, where existential quantifiers are guarded by negated transitive atoms (no restrictions are imposed on the usage of universal quantifiers, cf.~Section~\ref{sec:prel} for precise definition).  

It has already been mentioned that \fod\ has the finite
model property. Adding one transitive relation to \gfd\ (even
restricted to guards) we can write infinity axioms, i.e.~formulas that have only infinite models. However, models for this logic still enjoy the so-called tree-like
property, i.e.~new elements required by $\forall\exists$-conjuncts
can be added  independently. This property does no longer hold for \fodt{} or the fragments mentioned above, where one can write arbitrary universal formulas with two variables (cf.~Section~\ref{sec:nontr-wit} for some examples). 

This article was originally meant to be a full version of the conference paper \cite{ST13} where we announced the theorem that \fodt\ is decidable. In the meantime, we have realized that one of the technical lemmas of the conference paper is flawed (Claim~10, page~323) and the technique introduced there gives decidability only for \fodt{} with transitive witnesses. The main result of this article is that  the satisfiability problem for \fodt{} with transitive witnesses is decidable in \TwoNExpTime. We also discuss limitations of our technique and point out why it does not extend to give decidability of the satisfiability problem for \fodt{} with free witnesses.  Accordingly, the status of the satisfiability problem for \fodt{} remains open. 

\begin{figure}[tbh]
\begin{center}
\small\hspace{-2mm}  
\begin{tabular}{|c||c|c|c|}\hline
 & \multicolumn{3}{|c|}{} \\
{\large{\em Logic}  }  &
\multicolumn{3}{c|}{\large{ \em with transitive relations:}} \\
\cline{2-4}
 &  1 & 2 & 3 or more \\
%
%
\hline &  \multicolumn{3}{c|}{}
\\
\GFt{}+TG &  \multicolumn{3}{c|}{\TwoExpTime-complete}\\
&    \multicolumn{3}{c|}{\cite{ST04, Kie06} }\\
 \hline
 & & &
\\
\GFt{} & \TwoExpTime-complete & undecidable & undecidable \\

  & \cite{Kie05}  &  \cite{Kie05} & \cite{GMV99}\\

\hline
& & &
\\
 \FOt{} & {{{ SAT: ??? }}}  & undecidable
&
undecidable \\
&  {FinSAT: in 3-\NExpTime{}  } & \cite{Kie05,Kaz06} & \cite{GradelOR99}\\
&  \cite{P-H18} & & \\
 \hline

\multicolumn{4}{c}{}\\
\multicolumn{4}{c}{\large \em \hspace{2cm} with linear orders:}\\
\hline
& & &
\\
\FOt{} & \NExpTime-complete & SAT: ??? & undecidable \\
&   {\cite{Otto01}}  & FinSAT: in \TwoNExpTime{}  & \cite{Kie2011}\\
&     &  \cite{ZeumeH2016} & \\
\hline

\multicolumn{4}{c}{}\\
\multicolumn{4}{c}{\large \em \hspace{2cm} with equivalence relations:}\\
\hline
& & & 
\\
\GFt{} & \NExpTime-complete & \TwoExpTime-complete & undecidable \\
 
& {\tiny inherited from full \fod}  &  \cite{Kie05,KP-HT17} & \cite{KO12}\\

\hline & & &
\\
\FOt{} & \NExpTime-complete &
{{\TwoNExpTime}-complete} & undecidable \\
&   {\cite{KO12}}  & \cite{KO12,KMP-HT14}  & \cite{GO98}\\
&     &  & \\
\hline

\end{tabular}
\end{center}
\caption{\textsf{Two variable logics over transitive, linearly ordered or
equivalence structures}}\label{tabela}
\end{figure}

It should be pointed out that decidability of the finite satisfiability  problems for \fodt{} and for \fod{}  with two linear orders has already been confirmed and the following upper bounds are known: 3-\NExpTime{} for \fodt{} \cite{P-H18}, and \TwoNExpTime{} for \fod{} with two linear orders \cite{ZeumeH2016} (these bounds are not yet known to be tight). 
 
It also makes sense to consider more expressive systems in which
we may refer to the transitive closure of some relation. In fact,
relatively few decidable fragments of first-order logic with
transitive closure are known. One exception is the logic \gfd\
with a transitive closure operator applied to binary symbols
appearing only in guards~\cite{Michaliszyn09}. This fragment
captures the two-variable guarded fragment with transitive guards,
\gfdtg, preserving its complexity. Also decidable is the prefix class $\exists\forall$ 
(so essentially a fragment of \fod) extended by the positive deterministic transitive
closure of one binary relation, which is shown to enjoy
the exponential model property \cite{ImmermanRRSY04}.  In   \cite{KM12}
it was shown that the satisfiability problem for the
two-variable universal fragment of first-order logic with
constants remains decidable when extended by the transitive closure of a
single binary relation. Whether the same holds for
the finite satisfiability problem is open.

Also of note in this context is the interpretation of \FOt{} over {\em
data words} and {\em data trees} that appear e.g.~in verification
and XML processing. Decidability of \FOt{} over data words with
one additional equivalence relation was shown in \cite{BDM-LICS06}. For more results related to \FOt{} over data words or data trees see e.g.~\cite{Manuel10, SchZ10,DavidLT10, 	NiewerthS11,BDM-PODS06}.

\noindent {\bf Outline of the proof.}
Models for \fodt-formulas, taking into account the interpretation
of the transitive relation, can be seen as partitioned into cliques (for a formal definition of a clique see Subsection~\ref{cliques}).
As usual for two-variable logics, we first establish a
``Scott-type'' normal form for \fodt{} that allows us to
restrict the nesting of quantifiers to depth two.  
Moreover, the form of the $\forall\exists$-conjuncts enables us to distinguish witnesses required inside cliques (i.e.~realizing a 2-type containing both $Txy$ and $Tyx$, cf.~Subsection \ref{cliques}) from witnesses outside cliques.
We also establish a {\em small clique property} for \fodt{}, allowing us to restrict attention to models with cliques exponentially bounded in the size of the signature. 
Further constructions proceed on the level of cliques rather than individual elements.

Crucial to our decidability proof for \fodt{} with transitive witnesses is the following property: any infinitely
satisfiable sentence has an infinite {\em narrow} model,
 i.e.~a model whose universe can be
partitioned into segments (i.e.~sets of cliques) $S_0,
S_1,\ldots$, each of doubly exponential size, such that every
element in $\bigcup_{k=0}^{j-1} S_k$ requiring a witness  outside
its clique has the witness either in $S_0$ or in $S_j$
(cf.~Definition~\ref{def-narrow}). As substructures of a model preserve satisfiabilty of universal sentences, this immediately implies that, when
needed, every single segment $S_j$ ($j>0$) can be removed from the
universe of the model, and the structure restricted to the  remaining part of the universe is also a (more regular) model. In particular, this property allows us to build certain {\em regular} 
models where every two segments of the infinite partition (except the
first) are isomorphic (cf.~Definition~\ref{def:canon}).
Moreover, in regular models  the connection types between segments
can be further
simplified to a two element set. This
construction is based on the infinite Ramsey
theorem \cite{Ram30}, where
segments of the models correspond to nodes in a colored
graph, and connection types between segments correspond to colors of
edges.

The above properties suffice to obtain a \TwoNExpTime{} decision
procedure for the satisfiability problem for \fodt{} with transitive witnesses. We note that the best lower bound  coming from \gfdtg\ is \TwoExpTime{} and our result leaves a  gap in complexity. We also point out that our
decision procedure cannot be straightforwardly generalized to 
solve the satisfiability problem for \fodt{} with free witnesses. 

\medskip The rest of the paper is structured as follows. In Section~\ref{sec:prel} we introduce the basic notions, define the normal form for \fodt{} and show the small clique property. 
In Section~\ref{sec:narrow} we introduce the notion of a {\em splice} and the notion of a {\em narrow model} -- the most important technical notions of the paper. 
In Section~\ref{sec:main} we give the main result of the paper. Section~\ref{sec:nontr-wit} contains a discussion of the limitations of the above technique; in particular we give an example of an \fodt{-}formula with free witnesses that has only infinite models but does not admit narrow models in the above sense. 

\section{Preliminaries}\label{sec:prel}

\subsection{Basic concepts and notations}
We denote by \fod\ the two-variable fragment of first-order logic
(with equality) over relational signatures. By \fodt, we
understand the set of \fod-formulas over any signature
$\sigma=\sigma_0 \cup \{T\}$, where $T$ is a distinguished binary
predicate. The semantics for \fodt\ is as for \fod, subject to
the restriction that $T$ is always interpreted as a
\emph{transitive} relation.

In this paper, $\sigma$-structures are denoted by Gothic capital
letters and their universes by corresponding Latin capitals. Where
a structure is clear from context, we frequently equivocate
between predicates and their realizations, thus writing, for
example, $R$ in place of the technically correct $R^\fA$. If $\fA$
is a $\sigma$-structure and $B\subseteq A$, then
$\fA\obciety B$ denotes the (induced) substructure of $\fA$ with the
universe $B$.

An (atomic, proper) {\em $k$-type} (over a given signature) is a maximal
consistent set of atoms or negated atoms over $k$ distinct variables $x_1,\ldots,x_k$, containing the atoms  $x_i\neq x_j$ for every pair of distinct variables $x_i$ and $x_j$. 
If $\beta(x,y)$
is a $2$-type over variables $x$ and $y$, then   $\beta\obciety x$  (respectively,
$\beta\obciety y$) denotes the unique $1$-type that is
obtained from $\beta$ by removing atoms with the variable $y$
(respectively, the variable $x$). We denote by $\AAA$ the set of
all $1$-types and by $\BBB$ the set of all $2$-types (over a given
signature). Note that $|\AAA|$  and $|\BBB|$ are bounded
exponentially in the size of the signature. We often identify a
type with the conjunction of its elements.

For a given $\sigma$-structure $\fA$ and $a \in A$ we say that $a$
\emph{realizes} a $1$-type $\alpha$ if $\alpha$ is the unique
$1$-type such that $\fA \models \alpha[a]$.  We denote by
$tp^{\fA}[a]$ the $1$-type realized by $a$. Similarly, for
distinct $a,b \in A$, we denote by $tp^{\fA}[a,b]$ the unique
$2$-type \emph{realized} by the pair $a,b$, i.e.~the 2-type
$\beta$ such that $\fA \models \beta[a,b]$.

Assume $\fA$ is a $\sigma$-structure and $B,C\subseteq A$. We
denote by $\AAA^\fA$ (respectively, $\AAA^\fA[B]$) the set of all $1$-types
realized in $\fA$ (respectively, realized in $\fA \obciety B$), and by
$\BBB^\fA$ (respectively, $\BBB^\fA[B]$) the set of all $2$-types realized in
$\fA$ (respectively, realized in  $\fA \obciety B$). We denote by $\BBB^\fA[a,B]$
the set of all $2$-types $tp^{\fA}[a,b]$ with $b \in B$, and by
$\BBB^\fA[B,C]$ the set of all $2$-types $tp^{\fA}[b,c]$ with
$b \in B,  c\in C$.

Let $\gamma$ be a $\sigma$-sentence  of the form $\forall x\,
\exists y\,  \psi (x,y)$ and $a\in A.$ We say that an element
$b\in A$ is a {\em $\gamma$-witness} for $ a$ in the structure
$\mathfrak A$
 if $\mathfrak{A} \models \psi[a,b]$; $b$ is a {\em proper $\gamma$-witness},
 if $b$ is a $\gamma$-witness and $a\not= b.$

\subsection{Scott normal form}

As with \FOt{}, so too with \fodt, analysis is facilitated by the
availability of normal forms.

\setcounter{equation}{0}

\begin{definition}
\label{normalform}
An \fod-sentence $\Psi$ is in {\em Scott normal form} if it is  of
the following form:
\begin{eqnarray*}
\label{f1} &&\forall x \forall y\,  \psi_0 (x,y) \wedge
\bigwedge_{i=1}^M \forall x\exists y\, \psi_i(x,y),
\end{eqnarray*}
where every $\psi_i$ is quantifier-free and includes unary and
binary predicate letters only.
\end{definition}

In the above normal form without loss of generality we suppose that for $i\geq1$, $\psi_i(x,y)$ entails
$x\neq
 y$ (replacing $\psi_i(x,y)$ with $(\psi_i(x,y) \vee \psi_i(x,x))
 \wedge x\neq y$, which is sound over all structures
 with at least two elements).

Two formulas are said to be {\em equisatisfiable} if they
are satisfiable over the same universe. The following lemma is
typical for two-variable logics.

\begin{lemma}[\cite{Sco62,GKV97}]
\label{l-normalform} For every formula $\phi\in\FOt$ one can
compute in polynomial time an equisatisfiable, normal form
formula $\psi\in\FOt$ over a new signature whose length is linear
in the length of $\phi$.
\end{lemma}

Suppose the signature $\sigma$ consists of predicates of arity at most 2. To define  a $\sigma$-structure $\fA$, it suffices to specify the 1-types  and 2-types realized by elements and pairs of elements from the universe $A$.
In the presence of a transitive relation,
we classify 2-types according to the transitive
connection between $x$ and $y$. And so, we distinguish $\BBBr=\{\beta: Txy \wedge \neg
Tyx\in \beta\}$, $\BBBl=\{\beta: \neg Txy \wedge Tyx\in \beta\}$, $\BBBlr=\{\beta: Txy \wedge 
Tyx\in \beta\}$ and $\BBBn=\{\beta: \neg Txy \wedge \neg
Tyx\in \beta\}$.
Obviously $$\BBB=\BBBr\;
\dot{\cup}\; \BBBl\; \dot{\cup}\; \BBBlr\;
 \dot{\cup}\; \BBBn.$$
For a quantifier-free \fodt-formula $\phi(x,y)$ we use superscripts $^\rightarrow$, $^\leftarrow$,
 $^\leftrightarrow$ and $^-$ to define new formulas that explicitly specify  the transitive connection between $x$ and $y$. For
 instance, for a quantifier-free formula $\phi(x,y)\in\fodt$
 we let
$$\phir(x,y):=Txy\wedge \neg Tyx \wedge \phi(x,y).$$
This conversion of \fodt-formulas leads to the following
variant of the Scott normal form that we will employ in this paper.

\begin{eqnarray}
&\forall x \forall y\, \psi_0 \wedge  & \bigwedge_{i=1}^m \gamma_i
\wedge \bigwedge_{i=1}^{\overline{m}} \delta_i
\label{f-dwa}
\end{eqnarray}
 where $\gamma_i=\forall x\exists y\, \psi_i^{d_i}(x,y)$ with $d_i\in\{^\rightarrow,
^\leftarrow, ^-\}$, and $\delta_i=\forall x\exists y\,
\psi_i^{^\leftrightarrow}(x,y)$.

When a sentence $\Psi$ in the normal form  \eqref{f-dwa}
is fixed, we often write $\gamma_i\in \Psi$ to indicate that $\gamma_i$ is a  conjunct of $\Psi$ of the form $\forall x\exists y\, \psi_i^{d_i}(x,y)$.

\begin{lemma} \label{lemma_normalform}
Let $\phi$ be an \fodt{}-formula over a signature $\tau$.  One can
compute, in polynomial time, a formula $\Psi$ in normal form~\eqref{f-dwa}, over a signature $\sigma$
consisting of $\tau$ together with a number of additional unary
and binary predicates such that: (i) $\models \Psi \rightarrow \phi$; and (ii) every model of $\phi$ can be expanded to a model of $\Psi$. 
\end{lemma}
\begin{proof}[Sketch]
We employ the standard technique of renaming subformulas familiar
from \cite{Sco62} and \cite{GKV97}, noting that any formula $\exists y \psi$
is logically equivalent to
$ \exists y \psi^\rightarrow \vee \exists y \psi^\leftarrow \vee
\exists y \psi^\leftrightarrow \vee \exists y \psi^-$.
\end{proof}

The following immediate observation is helpful when showing that a structure is a model of normal form formulas.

\begin{proposition} \label{claim-iff}
Assume $\mathfrak A$ is a $\sigma$-structure and $\Psi$ is a
\fodt-sentence over $\sigma$ in normal form (\ref{f-dwa}).
 Then ${\mathfrak
A}\models \Psi$ if and only if
\begin{enumerate}[(a)]
\item\label{A} for each $a\in A$, for each $\gamma_i$  $(1\leq i\leq m)$ there is a
$\gamma_i$-witness for $a$ in $\mathfrak A$,
\item\label{B} for each $a\in A$, for each $\delta_i$  $(1\leq i\leq \overline{m})$ there is a
$\delta_i$-witness for $a$ in $\mathfrak A$,
\item\label{C} for each $a,b\in A$, $tp^{\mathfrak A}
[a,b]\models \psi_0(x,y)$ and $tp^{\mathfrak A}
[a]\models \psi_0(x,x)$,
\item\label{D} $T^\fA$ is transitive in $\mathfrak A$.
\end{enumerate}
\end{proposition}

\subsection{A small clique property for \fodt}\label{cliques}

Let $\fA$ be a $\sigma$-structure and $T$ a transitive relation on $A$. A subset $B$ of $A$ is called
$T$-\emph{connected} if $\BBB[B]\subseteq \BBBlr[\fA]$. Maximal
$T$-connected subsets of $A$ are called \emph{cliques}. So, in this paper by a clique we always mean {\em a maximal clique}.  Note that
if $\fA \models \neg Taa$, for some $a\in A$, then
$\{a\}$ is a clique. It is obvious that the set of cliques forms a partition of $A$. If $B$ and $C$ are distinct cliques and $b, c\in A$ are distinct elements, then we write 
\begin{itemize}
	\item $b\mna c$ iff $Tbc$ but not $Tcb$,
	\item $b\mna C$ iff for all $c\in C$: $b\mna c$,
	\item $C \mna b$ iff for all $c\in C$: $c \mna b$,
	\item $B\mna C$ iff for all $b\in B$: $b\mna C$.
\end{itemize}
When it is not ambiguous we simply write $\mn$ instead of $\mna$.
It is routine to show:
\begin{lemma}\label{l:partialorder}
Let $\fA$ be a $\sigma$-structure. The interpretation of  $T$ is  transitive in $\fA$ iff the relation $\mna$ is a partial order on the set of cliques of $\fA$. 
\end{lemma}

Now we establish the {\em small clique property} for \fodt.

\begin{lemma} \label{lma:theorem_si} 
	Let $\Psi$ be a satisfiable \fodt-sentence in normal form, over a
signature $\sigma$. Then there exists a model of $\Psi$ in which
the size of each clique is bounded exponentially in $|\sigma|$.
\end{lemma}

We first show how to replace a single clique in models of
normal-form \fodt-sentences by an equivalent small one. The idea
is not new (cf.~\cite{GKV97}). It was used in \cite{ST04} to show
that $T$-cliques in models of \gfdtg\ can be replaced by
appropriate small structures called $T$-petals (Lemma 17).
Later, in~\cite{KO12} it was proved that for any structure $\fA$
and its substructure $\fB$, one may replace $\fB$ by an
alternative structure $\fB'$ of a bounded size in such a way that
the obtained structure $\fA'$ and the original structure $\fA$
satisfy exactly the same normal form \FOt{} formulas. Below we
present a precise statement of the latter lemma.

\begin{lemma}[\cite{KO12}, Prop.~4]
\label{lem:ssp} Let $\fA$ be a $\tau$-structure, $\fB = \fA\restr
B$ for some $B \subseteq A$, $\overline{B}:=A\setminus B$. Then there is a
$\tau$-structure $\fA'$ with universe $A' = B' \,\dot{\cup}\, \overline{B}$
for some set $B'$ of size bounded polynomially in $|\BBB[\fA]|$
such that
\begin{enumerate}[(i)]
\item
$\fA'\restr \overline{B} = \fA \restr \overline{B}$.
\item 
$\AAA[B'] = \AAA[B]$, whence $\AAA[\fA'] = \AAA[\fA]$;
\item 
$\BBB[B'] = \BBB[B]$ and $\BBB[B', \overline{B}] = \BBB[B, \overline{B}]$, whence
$\BBB[\fA'] = \BBB[\fA]$;
\item 
for each $b' \in B'$ there is some $b \in B$ with $\BBB[b',A']
\supseteq \BBB[b,A]$;
\item 
for each $a \in \overline{B}$: $\BBB[a,B'] \supseteq \BBB[a,B]$.
\end{enumerate}
\end{lemma}

The above lemma applies to arbitrary structures.  We strengthen the lemma to structures with a distinguished transitive relation and show how to replace a single clique in models of normal-form \fodt-sentences by an equivalent small one. 

\begin{lemma}\label{lem:sc}
Let $\fA$ be a $\sigma$-structure, $B \subseteq A$ be a clique in
$\fA$, $\overline{B}:=A\setminus B$.
Then there is a $\sigma$-structure $\fA'$ with universe
$A'= B' \,\dot{\cup}\, \overline{B}$ for some clique
$B'$ 
with $|B'|$ bounded exponentially in $|\sigma|$, such that
(i)-(iv) are as in Lemma \ref{lem:ssp}, and (v) is strengthened
to:
\begin{enumerate}[(i)]
\item[(v')] 
for each $a \in \overline{B}$, $\BBB[a,B'] \supseteq \BBB[a,B]$, and\\
if $\BBB[a,B]\subseteq \BBB^d[\fA]$, then $\BBB[a,B'] \subseteq
\BBB^d[\fA']$, for every $d\in \{^\rightarrow, ^\leftarrow,^-\}$.
\end{enumerate}
\end{lemma}
\begin{proof}
If $|B|=1$, then we simply put $B'=B$ and we are done. Otherwise,
let $\Ul$, $\Ur$, $\Ulr$ and $\Un$ be fresh unary predicates.  Let
$\bar{\fA}$ be the expansion of $\fA$ obtained by setting:
\begin{itemize}
	\item[--] $\Ulr$ true for all elements of $B$, 
	\item[--] $\Ur$ true for all elements $a\in
\overline{B}$ such that $a\mna B$,  
	\item[--] $\Ul$ true for all
elements $a\in \overline{B}$ with $B\mna a$
and 
	\item[--] $\Un$ true
for all elements $a\in \overline{B}$ with $a\sima B$.
\end{itemize}
Let the result of the application of Lemma \ref{lem:ssp} to
$\bar{\fA}$ and the substructure induced by $B$ be a structure
$\bar{\fA}'$, in which $B'$ is the replacement of $B$. By $\fA'$
we denote the restriction of $\bar{\fA}'$ to the original
signature, i.e.~the structure obtained from $\bar{\fA}'$ by
dropping the interpretations of $\Ul, \Ur, \Ulr$ and $\Un$. Then $\fA'$
is a structure with universe $\overline{B} \cup B'$ and $|B'|$ is
exponentially bounded in the signature.
\end{proof}

\begin{proof}[Proof of Lemma~\ref{lma:theorem_si}] We first argue
that the structure obtained as an application of
Lemma~\ref{lem:sc} satisfies the same normal form sentences over
$\sigma$ as the original structure. Let $\Psi$ be a sentence in
normal form over $\sigma$, $\fA \models \Psi$, $B \subseteq
A$ be a clique, 
$\overline{B}=A\setminus B$, and $\fA'$ with universe $A'= B' \,\dot{\cup}\,
\overline{B}$ be a result of application of Lemma~\ref{lem:sc} to $\fA$.

Observe that the formula $\forall x\forall y \psi_0$ is satisfied in
$\fA'$ thanks to property (iii) of Lemma~\ref{lem:sc}. Now, for
any $c \in \overline{B}$, properties  (i) and (v) guarantee that $c$ has all
required witnesses. For any $b \in B'$, the same thing is
guaranteed by property (iv).
Properties (iii) and (v') guarantee that the obtained substructure
is a maximal $T$-connected set, so indeed it is a clique in the
new model.

Let $\fA$ be a countable $\sigma$-structure. Let $I_1$, $I_2,
\dots$ be a (possibly infinite) sequence of all cliques in 
$\fA$. Let $\fA_0=\fA$ and $\fA_{j+1}$ be the structure $\fA_j$
modified by replacing clique $I_{j+1}$ by its small replacement
$I'_{j+1}$ as described above. We define the limit structure
$\fA_\infty$ with the universe $\bigcup_{k=1}^\infty I'_k$ such that
for all $k, l$ the connections between $I'_k$ and $I'_l$ are
defined in the same way as in $\fA_{max(k, l)}$. It is easy to see
that $\fA_\infty$ satisfies $\Psi$ and all cliques in $\fA_\infty$
are bounded exponentially in $|\sigma|$.
\end{proof}

\subsection{Fragments of \fodt}
In this paper we consider two restrictions of \fodt{} depending on how the existential quantifiers are used; no restrictions are imposed on using universal quantifiers.

The fragment {\em with transitive witnesses}, \fodeg{}, consists of the formulas of \fodt{} where, when written in negation normal form, existential quantifiers are 'guarded' by transitive atoms, i.e.~they are applied to formulas with two free variables only of the form $\xi(x,y)\wedge \psi$, where  $\xi(x,y)$ is one of the conjunctions: $Txy\wedge Tyx$,  $Txy\wedge \neg Tyx$, or $\neg Txy\wedge Tyx$, and $\psi\in \fodeg$. 

Similarly, the fragment {\em with free witnesses}, \fodeng{}, consists of these formulas where, when written in negation normal form, 
existential quantifiers are applied to formulas with two free variables only of the form $\neg Txy\wedge \neg Tyx\wedge \psi$ with $\psi\in \fodeng$. We will return to this fragment in Section \ref{sec:nontr-wit}.

It is clear that in the normal form formula obtained by Lemma~\ref{lemma_normalform} for an \fodeg-formula all existential conjuncts have the form $\forall x\exists y\, \psi_i^{d_i}(x,y)$ with $d_i\in\{^\rightarrow,
^\leftarrow, ^\leftrightarrow\}$. Similarly, after transforming an \fodeng-formula into a normal form all existential conjuncts have the form $\forall x\exists y\, \psi_i^{-}(x,y)$.

\section{Narrow models of sentences of \fodeg{}}\label{sec:narrow}

In this section we assume $\Psi$ is an \fodeg-sentence in normal form, 
$\Psi = \forall x \forall y\, \psi_0 \wedge   \bigwedge_{i=1}^m \gamma_i
\wedge \bigwedge_{i=1}^{\overline{m}} \delta_i$,
where each $\gamma_i=\forall x\exists y\, \psi_i^{d_i}(x,y)$ with $d_i\in\{^\rightarrow,
^\leftarrow\}$, and each $\delta_i=\forall x\exists y\,
\psi_i^{^\leftrightarrow}(x,y)$. We also assume that the signature $\sigma$ consists of all relation symbols that appear in $\Psi$. 
Since \fodt{} satisfies the small clique property (Lemma \ref{lem:sc}) we additionally assume that whenever  $\fA\models\Psi$ then the size of each clique in $\fA$ is bounded exponentially in $|\sigma|$. 

Crucial to our decidability proof for \fodeg is the following property: any 
satisfiable sentence has a {\em narrow} model,
i.e.~either a finite model of bounded cardinality or an infinite model  whose universe can be
partitioned into segments (i.e.~sets of cliques) $S_0,
S_1,\ldots$, each of doubly exponential size, such that every
element from $\bigcup_{k=0}^{j-1} S_k$ requiring a witness  outside
its clique  has the witness either in $S_0$ or in $S_j$
(cf.~Definition~\ref{def-narrow}).

To prove existence of narrow models, we first introduce new technical notions and make some useful
observations. Then we show that from any narrow
model we can build a {\em regular model}
where the connection types between segments can be appropriately simplified. This finally leads to the algorithm deciding satisfiability for \fodeg{}. 

In the proof we employ the following property concerning extensions of partial orders that can be proved in a standard fashion. If $B$ and $C$ are distinct cliques and $b, c\in A$ are distinct elements, then we write 
\begin{itemize}
\item $b\sima c$ iff neither $b\mna c$ nor $c\mna b$,
\item $b\sima C$ iff for all $c\in C$,  $ b \sima c$,
\item $B\sima C$ iff  neither $B\mna C$ nor $C\mna B$.
\end{itemize}
In such a case we  say that elements $b$ and $c$ are {\em incomparable} (similarly for $B$ and $C$). When it is not ambiguous we simply omit the subscript $\fA$.

\begin{proposition}\label{prop:partial-extend}
	Let $R$ be a (strict) partial order on a set $A$, $a,b\in A$ and $a\sim b$. Let $C=\{c\in A: c R a \text{ or } c=a\}$ and $D=\{d\in A: b R d \text{ or } d=b\}$. 
	Then $R'=R \cup (C\times D)$ is a partial order on $A$, and it is a minimal partial order extending $R$ to contain $(a,b)$. 
\end{proposition}

\subsection{Splices}

In this section we analyse properties of models of $\Psi$ on the
level of cliques rather than individual elements.  The key technical argument of this section  is Corollary~\ref{wniosek-dla-wielu} saying, roughly speaking, that if
$\fA\models\Psi$ and elements of a finite subset $F$ of the
universe $A$ have their $\gamma_i$-witnesses in several
,,similar'' cliques (similar cliques realize the same {\em
splice}, see Definition~\ref{def-splice} below), then it is possible to
extend $\fA$ by  one such clique, where all the elements of $F$
have their $\gamma_i$-witnesses.

First, we need to introduce some new notions and notation. 

Let $\fA$ be a $\sigma$-structure. For
$a\in A$ denote by $Cl^{\fA}(a)$ the unique clique $C\subseteq A$
with $a\in C$. For $F\subseteq A,$ denote
$Cl^{\fA}(F)=\{Cl^{\fA}(a):\, a\in F\}$ and finally,
$Cl^{\fA}=Cl^{\fA}(A)$.

\begin{definition}\label{def-splice} Let $\fA$ be a $\sigma$-structure and $B\in Cl^{\fA}$.  
	An {\em $\fA$-splice} is a
	triple $\spl^{\mathfrak A}(B)=\langle tp^\fA[B], \In^{\fA}(B),
	\Out^{\fA}(B)\rangle$, where 
\begin{itemize}
	\item $tp^\fA[B]$ is the isomorphism type of the substructure induced by $B$,
	\item
	$\In^{\fA}(B)\stackrel{def}{=}\{tp^{\fA}[a]: a\in A\mbox{ and }
	a\mna B\},$
	\item
	$\Out^{\fA}(B)\stackrel{def}{=}\{tp^{\fA}[a]: a\in A\mbox{ and }
	B\mna a \}.$
\end{itemize}

	 Denote 
	$\SP^{\mathfrak A}$ as the set of all $\fA$-splices,
	$$\SP^{\fA}=\{\spl^{\fA}(B): B\in Cl^{\fA}\}.$$
\end{definition}
Note that the number of types of cliques is exponential w.r.t. the bound on the size of each clique in $\fA$, and the number of subsets of $\AAA$ is doubly exponential in $|\sigma|$. Hence, in any structure with the small clique property 
given by Lemma \ref{lma:theorem_si}, the number $s=|Sp^{\fA}|$ is bounded 
doubly exponentially in $|\sigma|$.  Define
\begin{equation}
{\mathbb M}= (m+1)\cdot s^2\cdot h,\label{def:M}
\end{equation}
where $m$ is the number of $\gamma_i$s in $\Psi$ and $h$ is the bound on the size of each
clique in $\fA$ given by Lemma~\ref{lma:theorem_si}.

In any $\sigma$-structure we distinguish the set ${\mathbb K}(\fA)\subseteq Cl^\fA$ of cliques with unique splices:
$${\mathbb K}(\fA)=\{B\in Cl^{\fA}\!: \mbox{ there is no }C\in  Cl^{\fA}\mbox{ with } \spl^\fA(B)\!=\!\spl^\fA(C) \mbox{ and } B\!\neq\! C\}$$
\noindent and
the corresponding subset ${K}(\fA)\subseteq A$ consisting of 
elements of the distinguished cliques\footnote{Readers familiar with the classical papers on \fod{} might think that ${K}(\fA)$ is a  {\em royal} part of the model (cf. for example \cite{Mor75} or \cite{GKV97}).}:
$${K}(\fA)=\bigcup_{B\in{\mathbb K}(\fA)}B.$$

For every conjunct $\gamma_i$ of $\Psi$ and for every $a\in A$ we define
${W}_i^{\fA}(a)$ as the set of all proper $\gamma_i$-witnesses for
$ a$ in a structure $\mathfrak A:$
$${W}_i^{\fA}(a) \stackrel{def}{=} \{b\in A: \fA
\models \psi_i^{d_i}[a,b]\mbox{ and } a\neq b\}.$$

As announced above, our goal in this section is to show existence of narrow models as defined below.

\begin{definition}\label{def-narrow}
	A model $\fA$ of $\Psi$ is  {\em narrow} if $A= K(\fA)$ or there is an infinite
	partition $P_A=\{S_0, S_1,\ldots\}$ of the universe $A$ such that $K(\fA) \subseteq S_{0}$ and for every $j\geq 0$:
	\begin{enumerate}
		\item $|S_j|\leq {\mathbb M}$, 
		\item  for every $a\in \bigcup_{k=0}^j S_k$
		and for every $\gamma_i\in \Psi,$
		
		\hspace{0,5cm}if ${W}_i^{\fA}(a)\cap S_0=\emptyset$, then
		${W}_i^{\fA}(a)\cap S_{j+1}\neq \emptyset.$
	\end{enumerate}
\end{definition}

\begin{lemma}\label{lem-narrow}
	Every 
	satisfiable \fodeg-sentence $\Psi$ has a narrow model.
\end{lemma}
The proof of the above lemma is deferred to Subsection~\ref{sec:proof-of-narrow-lemma}. 
Below we first introduce the notion of a {\em witness-saturated} model, in which every element requiring a witness outside ${K}(\fA)$ has infinitely many such witnesses, and we show existence of such models. In Subsection~\ref{sec:compression} we present the main technical tool that we later apply to prove Lemma~\ref{lem-narrow}.

\subsection{Saturated models}\label{sec:saturated}

In the definition below we introduce the notion of a witness-saturated model of a sentence $\Psi$. Any witness saturated
model that consists not just of $K(\fA)$ must be infinite (while $\Psi$ may have
other finite models). More technically, all the sets of witnesses outside of $K(\fA)$ are infinite.

\begin{definition}
	Assume $\fA\models\Psi.$ We say that $\fA$ is {\em
		witness-saturated}, if $\fA$ has the small clique property and for
	every $a\in A$, for every $\gamma_i\in \Psi$ $(1\leq i\leq m)$
	$${W}_i^\fA(a) \subseteq {K}(\fA) \mbox{\quad or \quad}
	{W}_i^\fA(a) \mbox{ is infinite}.$$
\end{definition}
Observe that it is possible that a witness-saturated model $\fA$ is finite -- then $A=K(\fA).$ 

The main and rather obvious  property of a  witness-saturated model is that a finite subset of cliques of such a model is almost always {\em redundant}. 

\begin{definition}\label{def-refundant}
	Let $S\subseteq A$ be a finite subset of $A$ such that $S\cap K(\fA)=\emptyset$. We say that $S$ is a {\em segment} in $\fA$ if for every $a\in
	S$: $Cl^\fA(a)\subseteq S.$
	A segment $S\varsubsetneq A$ is {\em redundant in} $\fA$, if for
	every $a\in A\setminus S$ and for every conjunct $\gamma_i$ of
	$\Psi$ we have:
	
	\hspace{0,5cm}  ${W}_i^{\fA}(a)\cap S\neq \emptyset$   implies there exists
	$c\in A\setminus S$ such that $c\in{W}_i^{\fA}(a).$
\end{definition}

\begin{proposition}\label{prop-redundant-first}
	If $\fA \models \Psi$, $\fA$ is witness-saturated  and  $S\varsubsetneq A $ is a  redundant segment
	in $\fA$, then $\fA\obciety (A\setminus S) \models
	\Psi.$
\end{proposition}

\begin{proof}
	Every subgraph of a transitive graph is also transitive.
	Conditions (\ref{A})--(\ref{C}) of Proposition \ref{claim-iff}
	obviously hold for $\fA\obciety (A\setminus S)$.
\end{proof}

\begin{lemma}[{\rm Saturated model}]
	\label{lemma-saturated} Every satisfiable normal form \fodeg-sentence $\Psi$ has a
	countable witness-saturated model. 
\end{lemma}

Before giving the proof of Lemma \ref{lemma-saturated} we  introduce some more notation. 
\begin{definition}\label{def-conn-transfered}              Let $\fA$ be a
	$\sigma$-structure, $B,C\subseteq A$  such that
	$B\cap C=\emptyset$. A {\em connection type} between $B$ and $C$ in $\fA$ is the structure
	$\langle B,C\rangle_\fA \stackrel{def}{=}\fA\obciety (B\cup
	C).$
	
 Let $\fA'$ be a $\sigma$-structure, $B',C'\subseteq A'$ and $f_B: B'\mapsto B$, $f_C: C'\mapsto C$ be isomorphisms between corresponding substructures. 
	We say that the connection type $\langle B',C'\rangle_{\fA'}$  in $\fA'$ {\em is transferred from} the connection type   $\langle B,C\rangle_\fA$ in $\fA$ w.r.t.~$f_B$ and $f_C$, denoted 
\begin{center}
	$\langle B',C'\rangle_{\fA'} \equiv_{f_B,f_C} \langle B,C\rangle_{\fA} \mbox{ if } $\\
	$\mbox{for every } b'\in B', c'\in C':\,\,tp^{\fA'}[b',c'] = tp^\fA[f_B(b'),f_C(c')].$
\end{center}
We denote by $\langle B',C'\rangle_{\fA'} :=_{f_B,f_C} \langle B,C\rangle_{\fA}$ an operation that {\em transfers} the corresponding connection type from $\fA$ to $\fA'$ by setting the 2-types $tp^{\fA'}[b',c']=tp^\fA[f_B(b'),f_C(c')]$ for every $b'\in B'$ and $c'\in C'$. 
\end{definition}
	
The operation $\langle B',C'\rangle_{\fA'} :=_{f_B,f_C} \langle B,C\rangle_{\fA}$ does not change anything else in any of the two structures; it might be used when $\fA'$ is not fully defined.  Whenever any of the isomorphisms is the identity function, we denote it by~$\id$. 

Lemma \ref{lemma-saturated} is a consequence of an iterative application of the following Claim \ref{claim-single-dupl}. It states that every  clique $B$ whose splice in a given model is non-unique can be properly
duplicated.  The copy, $D$, of $B$ is added in such a way that
it also provides, for all conjuncts of the form $\gamma_i$,  all
$\gamma_i$-witnesses for elements outside both of the cliques $B$ and $D$, that have been provided by $B$. Assume $B_1$ is another clique with the same splice as $B$. The  construction of  $\exb$ (cf. Definition \ref{def:exb}) transfers the connection type between $B_1$ and $B$ to the connection type  between $D$ and $B$; and the connection type between  $D$ and  the rest of $\exb$ is transferred from the connection type between $B$ and the rest of $\fA$.

\begin{definition}[]\label{def:exb}
	Assume $\fA\models \Psi,$ $B, B_1\in Cl^\fA, B_1\neq B$ and $\spl^{\fA}(B_1)=\spl^{\fA}(B)$.  Define $\exb$ as an
	extension of $\fA$ in the following way.

 Let $\fD$ be a fresh copy of $\fA\obciety B,$ $D\cap
	A=\emptyset.$ Let $f:D\mapsto B$ and $f_1:D\mapsto B_1$ be  appropriate isomorphisms of cliques.  The universe of $\exb$ is   $A_{+B} = A\dot{\cup} D$ and:

		\begin{enumerate}[(i)]
		\item $\langle D,B\rangle_{\exb} :=_{f_1,\id} \langle B_1,B\rangle_{\fA}$,
		\item $\langle D,A\setminus B\rangle_{\exb} :=_{f,\id} \langle B,A\setminus B\rangle_{\fA}$.
	\end{enumerate}
\end{definition}
The above construction will be used several times in the sequel and below we summarize the crucial properties of the structure $\exb$. 

\begin{claim}[{\rm Simple duplicability}]
\label{claim-single-dupl} Let $\fA$, $B$, $B_1$ and $\exb$ be as given in Definition \ref{def:exb}. Then
\begin{enumerate}
\item
$\exb\models \Psi,$
\item for every conjunct $\gamma_i$ of $\Psi,$ for every $a\in A\setminus B$  we have: \\ \hspace*{1cm} if ${W}_i^{\fA}(a)\cap{B}\neq \emptyset \mbox{ \quad then \quad}
{W}_i^{\exb}(a)\cap{D}\neq \emptyset$,
\item $\spl^{\exb}(D)=\spl^{\exb}(B)=\spl^\fA(B)$,
\item  $\BBB^{\exb}[D,B]= \BBB^{\exb}[B_1,B] = \BBB^{\fA}[B_1,B]$, 
\item  $\BBB^{\exb}[D,X]= \BBB^{\exb}[B,X] = \BBB^{\fA}[B,X]$ for each $X\in Cl^{\exb}, X\neq B, X\neq D.$
\end{enumerate}
\end{claim}

\begin{proof} Assume  $\fD$, $f$ and $f_1$ are as in Definition \ref{def:exb} (see Figure \ref{case2}, where $B_1\mn_{\fA}B$). First note that the clique $B$ does not change its splice in the structure $\exb$, namely $\spl^\exb(B)=\spl^\fA(B)$, as when defining 2-types between elements from the new clique $D$ and elements from $B$ (line (i) of Definition~\ref{def:exb}) no new 2-types are used.  

Condition (4) follows directly from line (i) of Definition~\ref{def:exb} 
and condition (5) follows from line (ii). 
It should be clear that when $B \sima B_1$ then $D \simb B$ and $D \simb B_1$, and then condition (3) holds. In the remaining cases, $\spl^\fA(B)=\spl^\fA(B_1)$ implies that $ \{tp^\fA[b]:b\in B\}\subseteq \In^\fA(B)$ and $\{tp^\fA[b]:b\in B\}\subseteq \Out^\fA(B)$, which furthermore implies
$$\In^\fA(B)\subseteq \In^\exb(D)\subseteq \In^\fA(B)\cup \{tp^\fA[b]:b\in B\}=\In^\fA(B)$$ 
and, similarly,
$$\Out^\fA(B)\subseteq \Out^\exb(D)\subseteq \Out^\fA(B)\cup \{tp^\fA[b]:b\in B\}=\Out^\fA(B).$$
  So, we obtain $\In^\exb(D)=\In^\fA(B)$,  $\Out^\exb(D)=\Out^\fA(B)$, and in fact $\spl^\exb(D)=\spl^\fA(B)$. This taking into account our first observation implies condition~(3) for all cases. 

Condition {(2)} also follows from Definition \ref{def:exb}: if $a\in A\setminus B$ has its witness in $B$, say $b\in {W}_i^{\fA}(a)\cap{B}$, then $f^{-1}(b)\in D$ is a witness of $a$ since $tp^{\exb}[f^{-1}(b),a]\stackrel{def}{=}tp^{\fA}[f(f^{-1}(b), a]=tp^{\fA}[b, a]$ as defined in line (ii).

It remains to prove (1). To see that  $\exb\models \Psi$ we show that conditions
(\ref{A})--(\ref{D}) of Proposition \ref{claim-iff} hold for $\exb$.

First, since $\exb$ is an extension of $\fA$, conditions (\ref{A})
and (\ref{B}) of Proposition \ref{claim-iff} hold for every $a\in
A$. Condition (\ref{B})
holds for $d\in D$ since cliques $D$ and $B$ are isomorphic. Finally, condition (\ref{A}) holds for $d\in D$, since it was true for  $f(d)\in B$ in $\fA$.

Secondly, by construction, $\BBB[\exb]=\BBB[\fA]$, hence condition
(\ref{C}) also holds.

So, it remains to show condition (\ref{D}), i.e.~that $T$ is
transitive in $\exb$. By construction, $T$ is transitive in $\exb\obciety A$ and in
$\exb\obciety{(A_{+B}\setminus B)}$.
We have three cases: $B_1\simw_{\fA} B,$ $B_1\mn_{\fA} B$ or $B_1\wi_{\fA} B$. 

\medskip\noindent{\sc Case 1.} $B_1\simw_{\fA} B.$

Then, by (4), $D\simw_{\exb}B$, and as $T$ is transitive in $\exb\obciety A$ and in $\exb\obciety
(A_{+B}\setminus B)$,  $T$ is also transitive in $\exb$.

\medskip\noindent{\sc Case 2.} $B_1\mn_{\fA}B$ (see Figure \ref{case2}).

\begin{figure}[htb]
	\hspace{3cm}
	\begin{tikzpicture}[scale=0.8]
	
	\draw (0,2) ellipse (1.5 and .5);
	\draw (8,2) ellipse (1.5 and .5);
	\draw[style=thick] (3,5) ellipse (1.5 and .5);
	
	\draw (4,3.2) ellipse (1.5 and .5);
	\coordinate [label=center:{\footnotesize$X$}] (A) at (5.8, 3.4);
	\draw[<-][style=thick] (1.5,2.3) -- (2.4, 3);
	\draw[-][style=dotted, very thick] (5.6,3) -- (6.8, 2.5); 
	\draw[->][style=double, thick] (5.3,3.6) -- (3.7, 4.3); 

	\coordinate [label=center:{\footnotesize$B$}] (A) at (-1.7, 2.1);
	\coordinate [label=center:{\footnotesize$B_1$}] (A) at (9.8, 2.1);
	\coordinate [label=center:{\footnotesize$D$}] (A) at (1.3, 5.3);
	
	\filldraw[fill=black] (3, 5) circle (0.04);
	
	\filldraw[fill=black] (1.6, 3.3) circle (0.04);
	\coordinate [label=right:{\footnotesize$a$}] (A) at (1.6, 3.3);
	\draw[-][style=dashed] (1.3, 3.1) -- (-0.4,2.1); 
	\filldraw[fill=black] (-0.5, 2) circle (0.04);
	\coordinate [label=left:{\footnotesize$b$}] (A) at (0, 2);
	\draw[-][style=dashed] (1.8, 3.5) -- (2.9, 4.8);
	\coordinate [label=left:{\footnotesize$f^{-1}(b)$}] (A) at (3, 5);
	
	\draw[<-][style=double,thick] (-1.3,2.6) -- (1.2, 4.5);
	\draw[->][style=double,thick] (9,2.6) -- (5, 4.5); 
	
	\draw[->][style=thick] (5.5,2) -- (2.2, 2); 
	
	\end{tikzpicture}
	\caption{\textsf{\footnotesize Case 2 of Claim \ref{claim-single-dupl}. Arrows depict ordering on cliques: single arrows depict the situation in $\fA$ ($B_1\mna B$); double arrows depict the ordering obtained after transferring connection types from $\fA$ to $\exb$ ($D\mnb B$ and  $B_1\mnb D$). The element $a\in A_{+B}$ has its $\gamma_i$-witnesses, $f^{-1}(b)$,  in $D$ as before in $B$ (dashed lines). A dotted line indicates that the respective connection type is arbitrary.  In case when  $X\mnb D$ we had also that  $X\mna B$ and then $\mnb$ is a partial order. }}\label{case2}
\end{figure}

By (4), $D\mnb B$. To show that $T$ is
transitive in $\exb$ we show that $\mnb$ is a partial order on the set of cliques in $\exb$ (cf.~Lemma~\ref{l:partialorder}). It suffices to consider the following three
subcases with $X \in Cl^{\exb}, X\neq B, X\neq D.$
\begin{itemize}
	\item $X\mnb D$ (and $D\mnb B$). \\
		Then, by (5), $X\mnb B$ as required. 
	\item  $D\mnb X$ and $X\mnb B$. \\
	This is impossible, since by (5) $D\mnb X$ implies $B \mnb X$. 
	\item ($D\mnb B$ and) $B\mnb X$.\\
	 Again by (5), $D\mnb X$, as desired. 
\end{itemize}

So, $T$ is transitive in $\exb$.

\medskip\noindent{\sc Case 3.} $B\mna B_1$.
This case is symmetric to the previous one.
\end{proof}

\begin{proof}[Proof of Lemma \ref{lemma-saturated}]
	Assume $\fA\models \Psi$, $\gamma_i\in \Psi$, $a\in A$ and 
	${W}_i(a) \nsubseteq {K}(\fA)$. Then there exists $B\in Cl^\fA$,
	$B\not\in {\mathbb K}(\fA)$ such that ${W}_i(a) \cap B \neq 		\emptyset$. 
	Let $B_1\in Cl^\fA, B_1\neq B$ and $\spl^{\fA}(B_1)=\spl^{\fA}(B)$
	(the set $B_1$ exists since $B\not\in {\mathbb K}(\fA)$). 
	Define $\fA_{(0)}=\fA$ and $\fA_{(j+1)}=\fA_{(j)}^{+B}$. Finally let
	$$\fA_{(\infty,\gamma_i,a)}= \bigcup_{j=0}^{\infty}\fA_{(j)}^{+B}.$$
	By Claim \ref{claim-single-dupl} obviously $\fA_{(\infty,\gamma_i,a)}\models \Psi$. Iterate the above procedure for every $a\in A$ and $\gamma_i\in \Psi$.
\end{proof}

\subsection{Duplicability}\label{sec:compression}
The key
technical lemma of the paper is Corollary  \ref{wniosek-dla-wielu} below. It says that when several ele\-ments $a_1,
a_2,\ldots,a_p$ of a model $\fA$ have $\gamma_i$-witnesses in
several distinguished cliques that realize the same splice, one
can extend $\fA$ by a single clique $D$ (realizing the
same splice) in which all of $a_1, a_2,\ldots,a_p$ have their
$\gamma_i$-witnesses. In the proof we essentially  employ the property of simple duplicability given above in Claim \ref{claim-single-dupl}.

We start with Claim \ref{claim-dla-dwoch} below, where we have the following situation. Cliques $C$ and $C_1$ realize the same splice, similarly cliques $B$ and $B_1$ realize the same splice, $B\neq B_1$. There are elements of the clique $C$ (and $C_1$) which have their $\gamma_i$-witnesses in the clique $B$ (and $B_1$, respectively). But it is possible that there are elements of $C_1$ which have no  $\gamma_i$-witnesses in $B$ (or elements of $C$ that have no $\gamma_i$-witnessses in $B_1$). We show that one can  extend the model $\fA$ by a clique $D$, a copy of $B$, as in Claim \ref{claim-single-dupl}, but additionally in such a way, that not only  elements of $A\setminus B$ (in particular elements of $C$), have their $\gamma_i$-witnesses in $D$  (as before in $B$), but also elements of $C_1$ (which had their $\gamma_i$-witnesses in $B_1\setminus B$ before) will have their $\gamma_i$-witnesses in $D$.  The cost we pay is in condition {\em (ii)}: {\em for every $a\in E\cup C$} (for any  finite subset $E$ of the the universe), while in Claim \ref{claim-single-dupl} we had: {\em for every $a\in A\setminus B$}, but this compromise is sufficient for our purposes and makes proofs easier. 

Below, for distinct cliques $C$ and $C_1$ we employ the following abbreviation:

\hspace{0,5cm} $ \bullet\,\,C \lsima C_1$ iff $C\mna C_1$ or $C\sima C_1$.

\noindent Let us emphasize that the requirement 	$C \lsima C_1$ in the statement of Claim \ref{claim-dla-dwoch} is important, since the role of $C$ and $C_1$ is not fully symmetric (cf.~conditions (ii) and (iii)).
\begin{claim}\label{claim-dla-dwoch}
	Assume $\fA\models\Psi$ is countable witness-saturated and $\gamma_i\in
	\Psi$. Let $B, B_1,C,C_1\in Cl^\fA$, $B_1\neq B$, $\spl(B_1)=\spl(B)$, $\spl(C_1) = \spl(C)$  and    
	$C \lsima C_1$.
	Additionally, assume  $C\not\in {\mathbb K}(\fA)$,  $W_i(C)\cap{B}\neq\emptyset$ and $W_i(C_1)\cap{B_1}\neq\emptyset$.  Let $E$ be a finite subset of $A$. 
Then, there exists an  extension $\fA_1$ of $\fA$ by a clique $D$ such that
\begin{enumerate}[(i)]
		\item $\fA_1\models \Psi$ and $\fA_1$ is  witness-saturated,
		\item for every $a\in E\cup C$:  \quad if \quad ${W}_i^{\fA}(a)\cap B\neq \emptyset \mbox{ \quad then \quad}
		{W}_i^{\fA_1}(a)\cap D\neq \emptyset,$
		\item
		for every $a\in C_1$:  \quad if \quad ${W}_i^{\fA}(a)\cap B_1\neq \emptyset \mbox{ \quad then \quad}
		{W}_i^{\fA_1}(a)\cap D \neq \emptyset,$
		\item $\spl^{\fA_1}(D)=\spl^{\fA_1}(B)= \spl^{\fA}(B)$.
\end{enumerate}
\end{claim}

\begin{proof}
Let $\fA$,  $C$, $C_1$, $B$, $B_1$, $E$ be as given above, and $\gamma_i=\forall x\exists y\, \psi_i^{d_i}(x,y)$ with
$d_i\in\{^\rightarrow, ^\leftarrow\}$ (recall, incomparable witnesses in \fodeg are not allowed). Since $\spl(B)=\spl(B_1)$ and $B\neq B_1,$ we have $B
\not\in {\mathbb K}(\fA)$  and $B_1 \not\in {\mathbb K}(\fA)$.  Note that the form of $\gamma_i$ implies that it is never possible that $C\sima B$ or $C_1\sima B_1$.

Before we proceed with the proof let us discuss the goal in more detail. 
Claim~\ref{claim-single-dupl} allows one to extend a given model for $\Psi$ by a copy of a clique that has a non-unique splice, and this is what we also want to do in this claim; however, we need to work more.  

Consider $\exb$ -- the  extension of $\fA$ as given in Definition \ref{def:exb}. Then, by Claim~\ref{claim-single-dupl}, the following holds 
 \begin{enumerate}
 	\item
 	$\exb\models \Psi,$
 	\item for every $a\in A\setminus B$
 	we have: \\ \hspace*{1cm} if ${W}_i^{\fA}(a)\cap B\neq \emptyset \mbox{ \quad then \quad}
 	{W}_i^{\exb}(a)\cap D\neq \emptyset.$
 	\item $\spl^{\exb}(D)=\spl^{\exb}(B)=\spl^\fA(B)$.
 \end{enumerate}
If $\fA$ is witness-saturated then also $\exb$ is witness-saturated. Hence, condition (1) implies (i) of our claim, condition (2) implies (ii) for any $E\subseteq A$, and condition (3) implies (iv) of our claim. 
However, condition {\em (iii)} is not ensured: 
by construction, the connection type $\langle D,C_1\rangle$ in $\exb$ is transferred from the connection type  $\langle B,C_1\rangle$ in $\fA$, and
it is not guaranteed that elements of $C_1$ had their $\gamma_i$-witnesses in $B$; it is possible even that $C_1\simb D$.

A similar problem appears when we extend $\fA$ by adding a simple duplicate of $B_1$. Namely, let $\exba$ be the (witness-saturated) extension of $\fA$ given by Definition \ref{def:exb}, where  $\langle D,B_1\rangle_{\exba} :=_{f,\id} \langle B,B_1\rangle_{\fA}$ and
 $\langle D,A\setminus B_1\rangle_{\exba} :=_{f_1,\id} \langle B_1,A\setminus B_1\rangle_{\fA}$. 
Again, condition {\em (ii)} of our claim is not ensured in  $\exba$,
as  $\cont{D,C}{\exba}\equiv_{f_1,\id}\cont{B_1,C}{\fA}$, 
and it is not guaranteed that elements of $C$ 
had their $\gamma_i$-witnesses in $B_1$. 

\medskip To prove our claim we need a compromise between the two approaches described above.  In particular, we restrict $E$ to finite subsets of $A$ as this is sufficient for our future purposes and yields a simpler proof. 
In the proof we consider several formal cases depending of the order-relationships between the cliques $C$, $C_1$, $B$ and $B_1$, as systematically listed below.
Since elements of $C$ and $C_1$  have their $\gamma_i$-witnesses in, respectively, $B$ and $B_1$, in each case the relationship between $C$ and $B$ implies the same relationship between $C_1$ and $B_1$.

\noindent Case 1. $C=C_1$ and

\hspace{1cm}1.1. $C\mn B$ or

\hspace{1cm}1.2. $B\mn C$.

\noindent Case  2. $C\mn C_1$ and $C\mn B$ and

\hspace{1cm}2.1. $C_1\mn B$ or

\hspace{1cm}2.2. $B\mn C_1$ or

\hspace{1cm}2.3. $C_1\simw B$.

\noindent  Case  3. $C\mn C_1$ and $B\mn C$.

\noindent Case  4. $C\simw C_1$ and $C\mn B$.

\hspace{1cm}4.1.  $C_1\mn B$ or

\hspace{1cm}4.2.  $C\mn B_1$ or

\hspace{1cm}4.3.  $C_1\simw B$ and  $C\simw B_1$.

\noindent  Case 5. $C\simw C_1$ and $B\mn C$ 

\hspace{1cm}5.1.  $B\mn C_1$ or

\hspace{1cm}5.2.  $B_1\mn C$ or

\hspace{1cm}5.3.  $C_1\simw B$ and  $C\simw B_1$. 

No other cases are possible, for example in Case 4 we do not  have a subcase with $B \mn C_1$, as this by transitivity would imply $C_1\mn C$. 

In most cases we start with $\exb$ and modify the structure to obtain $\fA_1$ in which elements of $C_1$ will have their $\gamma_i$-witnesses also in $D$. In some cases it is more convenient to start from $\exba$ and then modify the structure to obtain $\fA_1$ in which elements of $C\cup E$ will have their $\gamma_i$-witnesses also in $D$. 
Next steps  depend on the particular case and are labelled by the case (subcase) number. 
Note that condition (5) of Claim~\ref{claim-single-dupl} ensures that in $\exb$ the relationship between $C$ and the cliques  $B$ and $D$  is also determined by $\gamma_i$. Namely we have:

(*) if $d_i=^\rightarrow$, then $C \mn B$, $C_1 \mn B_1$ and $C \mn D$; and if $d_i=^\leftarrow$, then 
$B \mn C$, $B_1 \mn C_1$ and $D \mn C$.

Let us start with the easiest Subcase 1.1. 

\medskip
\noindent{\sc Subcase 1.1.} $C=C_1$ and $C\mnb B$ (see Figure \ref{piccasecase11}). \\
To satisfy condition {\em
(iii)}, we modify the connection type $\cont{C,D}{\exb}$ as follows (recall $f: D\mapsto B$ and $f_1: D\mapsto B_1$ are isomorphisms of cliques given by Definition \ref{def:exb}). 

\begin{figure}[htb]
	\begin{tikzpicture}[scale=0.8]
	
	\draw (-2,3.5) ellipse (1.2 and .5);
	\coordinate [label=center:{\footnotesize$B$}] (A) at (-3.2, 4.1);

	\draw (8.4,3.5) ellipse (1.5 and .5);
	\coordinate [label=center:\footnotesize{$B_1$}] (A) at (9.9, 4.1);
	\filldraw[fill=black] (8.5, 3.5) circle (0.04);
	\coordinate [label=left:{\footnotesize$f_1(d)$}] (A) at (9.8, 3.5);

	\draw (3,5) ellipse (1.5 and .5);
	\coordinate [label=center:{\footnotesize$C$}] (A) at (1.5, 5.6);
	\draw[->][style=thick] (0, 5) -- (-2.6, 4.2);	
	\draw[->][style=thick] (6, 5) -- (8.7, 4.2);
	\filldraw[fill=black] (4, 5) circle (0.04);
	\coordinate [label=left:{\footnotesize$c$}] (A) at (4, 5);
	\draw[->][style=dashed] (4, 4.8) -- (3.1, 2.2);
	\draw[->][style=dashed] (4.2, 5) -- (8.3, 3.6);
	\draw[->][style=double,thick] (1.7, 4.3) -- (1.1, 2.8);

	\draw (2.2,2)[style=thick] ellipse (1.5 and .5);
	\coordinate [label=center:{\footnotesize$D$}] (A) at (0.4, 2.4);
	\filldraw[fill=black] (3, 2) circle (0.04);
	\coordinate [label=left:{\footnotesize$d$}] (A) at (3, 2);
	
	\end{tikzpicture}
	\caption{\footnotesize\textsf{Subcase 1.1 in the proof of Claim \ref{claim-dla-dwoch}: $C=C_1$, $C\mn B$ and $D$ is a fresh copy of the clique $B$ given by Claim \ref{claim-single-dupl}.  Single arrows depict types from the model $\fA$:  $C\mnb B,$ and $C\mnb B_1$. 	Double arrow depicts the ordering obtained after transferring connection types from $\fA$ to $\exb$: $C\mnb D$. Dashed
			arrows are to show that corresponding pairs of  elements realize the same 2-type: $tp^{\exb}[c,d]$ is replaced by $tp^{\exb}[c,f_1(d)]$ in Step 1.1 to obtain $\fA_1$.}}\label{piccasecase11} 
\end{figure}

\medskip\noindent {\em Step 1.1. } 
For every $c\in C$, for every $d\in
D$, if $f_1(d)\in {W}_i^{\fA}(c)$, then 

\hspace{2cm}replace $tp^{\exb}[c,d]$
by $tp^{\exb}[c,f_1(d)]$. 

\medskip
In Step 1.1, every element $c\in C$ that has a $\gamma_i$-witness in $B_1$ (in the model $\exb$) is attributed a $\gamma_i$-witness in $D$. 
Note that $tp^{\exb}[f_1(d)]=tp^{\exb}[d]=tp^{\fA}[f(d)]$, $c\mnb d,$
and $c\mnb f_1(d),$  so no incompatibility occurs. Moreover, since only types of  $\BBBr$ are changed and they are replaced by types of $\BBBr$, after the modification $T$ remains transitive in $\fA_1$. 

Additionally, since no type $tp^{\exb}[c,a]$ with $c\in C$, $a\in A_{+B}\setminus D$ is changed, no witness for $a$ is stolen from $C$. The same is true for elements of $D$: after performing Step 1.1 every element $d\in D$ has a
$\gamma_j$-witness, for $j=1,2,\ldots, m$. This is because $\exb$
is witness-saturated, so when $tp^{\exb}[c,d]$ was replaced by
$tp^{\exb}[c,f_1(d)]$ and an element $c\in C$ was a
$\gamma_j$-witness for $d$  in the model $\exb$, then  $d$ had
another $\gamma_j$-witness in $A_{+B}$ (in fact infinitely many such witnesses as cliques are finite).

Observe that Subcase 1.2 is symmetric to 1.1. The only difference is in the order-relationship between the cliques $C$ and $D$ (certainly also between $C_1$ and $B_1$). \hfill {\em End of Case 1}  

\medskip

Before proving Subcase 1.1 we have listed systematically twelve (sub)cases that should be considered. The remaining ten (sub)cases can be organised into three groups denoted {\bf A}, {\bf B} and {\bf C}. The proofs for all cases from the same group proceed in the same way.  Each group is described by a condition defining a certain order-relationship between the cliques under consideration. 

\medskip
\noindent {\bf A.} {\em Cliques $C$ and $C_1$ are in the same order-relationship with $B$}.\\ This group consists  of four subcases (cf.~Figure \ref{fig:caseA}). The construction in each case starts from $\exb$ and then the connection type $\cont{C_1,D}{\exb}$ is modified in order to provide $\gamma_i$-witnesses for elements from $C_1$ in the newly added clique $D$. 
Below we describe the construction in detail for Subcase 2.1. The argument for remaining cases from this group is similar.

\begin{figure}[htb]
	\begin{tikzpicture}[scale=0.45]
	
	\coordinate [label=left:{\footnotesize$B$}] (A) at (-7.1, 3.5);
	\filldraw[fill=black] (-7, 3.5) circle (0.04);
	
	\filldraw[fill=black] (-4, 3.5) circle (0.04);
	\coordinate [label=right:{\footnotesize$B_1$}] (A) at (-4.1, 3.5);
	
	\filldraw[fill=black] (-7, 5.5) circle (0.04);
	\coordinate [label=left:{\footnotesize$C$}] (A) at (-7.1, 5.6);
	
	\filldraw[fill=black] (-4, 5.5) circle (0.04);
	\coordinate [label=right:{\footnotesize$C_1$}] (A) at (-4.1, 5.6);
	
	\draw[->][](-6.8,5.5) -- (-4.2, 5.5);
	\draw[->][](-7,5.4) -- (-7, 3.7);
	\draw[<-][](-6.8,3.6) -- (-4.2, 5.3);
	\draw[->][](-4, 5.4) -- (-4, 3.7);

	
	\coordinate [label=left:{\footnotesize$B$}] (A) at (0.1, 3.5);
	\filldraw[fill=black] (0, 3.5) circle (0.04);

	\filldraw[fill=black] (3, 3.5) circle (0.04);
	\coordinate [label=right:{\footnotesize$B_1$}] (A) at (2.9, 3.5);
	
	\filldraw[fill=black] (0, 5.5) circle (0.04);
	\coordinate [label=left:{\footnotesize$C$}] (A) at (0.1, 5.6);
	
	\filldraw[fill=black] (3, 5.5) circle (0.04);
	\coordinate [label=right:{\footnotesize$C_1$}] (A) at (2.9, 5.6);
	
	\draw[->][](0.2,5.5) -- (2.8, 5.5);
	\draw[<-][](0,5.4) -- (0, 3.7);
	\draw[<-][](3, 5.4) -- (3, 3.7);
	
	\coordinate [label=left:{\footnotesize$B$}] (A) at (7.1, 3.5);
	\filldraw[fill=black] (7, 3.5) circle (0.04);
	
	\filldraw[fill=black] (10, 3.5) circle (0.04);
	\coordinate [label=right:{\footnotesize$B_1$}] (A) at (9.9, 3.5);
	
	\filldraw[fill=black] (7, 5.5) circle (0.04);
	\coordinate [label=left:{\footnotesize$C$}] (A) at (7.1, 5.6);
	
	\filldraw[fill=black] (10, 5.5) circle (0.04);
	\coordinate [label=right:{\footnotesize$C_1$}] (A) at (9.9, 5.6);
	
	\draw[-][snake=zigzag,segment amplitude=1pt, very thin](7.2,5.5) -- (9.8, 5.5);
	\draw[->](7,5.4) -- (7, 3.7);
	\draw[->][](10, 5.4) -- (10, 3.7);
	\draw[<-][](7.2,3.6) -- (9.8, 5.3);
	
	\coordinate [label=left:{\footnotesize$B$}] (A) at (14.1, 3.5);
	\filldraw[fill=black] (14, 3.5) circle (0.04);
	
	\filldraw[fill=black] (17, 3.5) circle (0.04);
	\coordinate [label=right:{\footnotesize$B_1$}] (A) at (16.9, 3.5);
	
	\filldraw[fill=black] (14, 5.5) circle (0.04);
	\coordinate [label=left:{\footnotesize$C$}] (A) at (14.1, 5.6);
	
	\filldraw[fill=black] (17, 5.5) circle (0.04);
	\coordinate [label=right:{\footnotesize$C_1$}] (A) at (16.9, 5.6);
	
	\draw[->][snake=zigzag,segment amplitude=1pt, very thin](14.2,5.5) -- (16.8, 5.5);
	\draw[<-][](14,5.4) -- (14, 3.7);
	\draw[<-][](17, 5.4) -- (17, 3.7);
	\draw[->][](14.2,3.6) -- (16.8, 5.3);

	\coordinate [label=left:{\footnotesize subcase 2.1.}] (A) at (-3.1, 2.5);
	\coordinate [label=left:{\footnotesize  case 3.}] (A) at (3.1, 2.5);
	\coordinate [label=left:{\footnotesize subcase 4.1.}] (A) at (11, 2.5);
	\coordinate [label=left:{\footnotesize subcase 5.1.}] (A) at (18, 2.5);
	
	\end{tikzpicture}
	\caption{\footnotesize\textsf{Cases from group A in the proof of Claim \ref{claim-dla-dwoch}. Cliques $C$ and $C_1$ are in the same order-relationship with $B$.  Arrows depict the order-relationship between cliques, snake lines connect incomparable cliques.}}\label{fig:caseA}
\end{figure}

\medskip\noindent{\sc Subcase 2.1.} $C\mnb C_1$, $C\mnb B$, $C_1\mnb B$  (see Figure \ref{piccase21} ).
\\In this case by (*) we also have $C_1\mnb B_1$ and so, by transitivity, $C \mnb B_1$ and, by (*) we have $C_1\mnb D$. 
To obtain the model $\fA_1$ we modify the connection type $\cont{C_1,D}{\exb}$ as follows.

\begin{figure}[htb]
	\begin{tikzpicture}[scale=0.8]
	
	\draw (-2.2,3.5) ellipse (1.5 and .5);
	\coordinate [label=center:{\footnotesize$B$}] (A) at (-4, 4);
	
	\draw (8.3,3.5) ellipse (1.5 and .5);
	\coordinate [label=center:{\footnotesize$B_1$}] (A) at (9.8, 4);
	\filldraw[fill=black] (8.5, 3.5) circle (0.04);
	\coordinate [label=left:{\footnotesize$f_1(d)$}] (A) at (9.7, 3.5);
	\draw (0,5) ellipse (1.5 and .5);
	\coordinate [label=center:{\footnotesize$C$}] (A) at (-1.7, 5.6);
	\draw[->][style=thick] (-1.8, 5) -- (-3.6, 4.1);
	\draw[->][style=thick] (1, 5.7) -- (4, 5.7);
	\draw (4.8,5) ellipse (1.5 and .5);
	\coordinate [label=center:{\footnotesize$C_1$}] (A) at (6.8, 5.6);
	\draw[->][style=thick] (6.8, 5) -- (9.3, 4.1);
	\draw[->][style=double, thick] (4.6, 4.3) -- (4.1, 2.8);
	\draw[->][style=double, thick] (-0.9, 4.3) .. controls(0,3.3) .. (0.8, 2.8);
	\coordinate [label=left:{\footnotesize$c_1$}] (A) at (4, 5);
	\draw[->][style=dashed] (4, 4.8) -- (3.1, 2.2);
	\draw[->][style=dashed] (4.3, 5) -- (8.3, 3.6);
	\draw[->][style=thick] (2.8, 4.8) -- (0.1, 3.8);
	
	\filldraw[fill=black] (4, 5) circle (0.04);
	\draw [style=thick](2.3,2) ellipse (1.5 and .5);
	\coordinate [label=center:{\footnotesize$D$}] (A) at (0.2, 2.3);
	\filldraw[fill=black] (3, 2) circle (0.04);
	\coordinate [label=left:{\footnotesize$d$}] (A) at (3, 2);
	
	\end{tikzpicture}
	\caption{\footnotesize\textsf{Subcase 2.1 in the proof of Claim \ref{claim-dla-dwoch}: $C\mnb C_1$,  $C\mnb B$, $C_1\mnb B_1$, $C_1\mnb B$ (depicted by single arrows) and $D$ is a fresh copy of the clique $B$ given by Claim \ref{claim-single-dupl}. 
			Double arrow depicts the ordering obtained after transferring connection types from $\fA$ to $\exb$: $C_1\mnb D$. Dashed	arrows indicate that corresponding pairs of  elements realize the same 2-type: $tp^{\exb}[c,d]$ is replaced by $tp^{\exb}[c,f_1(d)]$ in Step 2.1 to obtain $\fA_1$.}}\label{piccase21} 
\end{figure}

\medskip\noindent {\em Step 2.1. }
\medskip\noindent For every $c_1\in C_1$, for every $d\in
D$, if $f_1(d)\in {W}_i^{\fA}(c_1)$, then 

\hspace{2cm}replace $tp^{\exb}[c_1,d]$
by $tp^{\exb}[c_1,f_1(d)]$. 

\medskip
In Step 2.1, every element $c_1\in C_1$ that has a $\gamma_i$-witness in $B_1$ (in the model $\exb$) is attributed a $\gamma_i$-witness in $D$.
Note that $tp^{\exb}[f_1(d)]=tp^{\exb}[d]$, $c_1\mnb d,$
and $c_1\mnb f_1(d),$  so no incompatibility occurs. Moreover, since only types of  $\BBBr$ are changed and they are replaced by types of $\BBBr$, after the modification $T$ remains transitive in $\fA_1$. 

Additionally, since no type $tp^{\exb}[c_1,a]$ with $c_1\in C_1$, $a\in A_{+B}\setminus D$ is changed, no witness for $a$ is stolen from $C_1$. The same is true for elements of $D$: after performing Step 2.1 every element $d\in D$ has a
$\gamma_j$-witness, for $j=1,2,\ldots, m$, as before in $\exb$. This is because $\exb$
is witness-saturated, so when $tp^{\exb}[c,d]$ was replaced by
$tp^{\exb}[c,f_1(d)]$ and an element $c\in C$ was a
$\gamma_j$-witness for $d$  in $\exb$, then another  $\gamma_j$-witness can be found in $A_{+B}$ (in fact infinitely many witnesses as cliques are of finite size).

\medskip	
\noindent {\bf B}. {\em Cliques $C$ and $C_1$ are in the same order-relationship with $B_1$.}\\
In this group we consider only cases that do not belong to group A, i.e.~we have also four subcases (cf.~Figure \ref{fig:caseB}).
The construction in each case from this group starts from $\exba$ that is then modified in order to provide $\gamma_i$-witnesses for all elements from $C\cup E$ in the newly added clique $D$. 
We describe the construction in detail for Subcase 2.2. The argument for remaining cases from this group is similar.

\begin{figure}[htb]
	\begin{tikzpicture}[scale=0.45]
	
	\coordinate [label=left:{\footnotesize$B$}] (A) at (-7.1, 3.5);
	\filldraw[fill=black] (-7, 3.5) circle (0.04);
	
	\filldraw[fill=black] (-4, 3.5) circle (0.04);
	\coordinate [label=right:{\footnotesize$B_1$}] (A) at (-4.1, 3.5);
	
	\filldraw[fill=black] (-7, 5.5) circle (0.04);
	\coordinate [label=left:{\footnotesize$C$}] (A) at (-7.1, 5.6);
	
	\filldraw[fill=black] (-4, 5.5) circle (0.04);
	\coordinate [label=right:{\footnotesize$C_1$}] (A) at (-4.1, 5.6);
	
	\draw[->][](-6.8,5.5) -- (-4.2, 5.5);
	\draw[->][](-7,5.4) -- (-7, 3.7);
	\draw[->][](-4, 5.4) -- (-4, 3.7);
	\draw[->][](-6.8,3.6) -- (-4.2, 5.3);
	
	\coordinate [label=left:{\footnotesize$B$}] (A) at (0.1, 3.5);
	\filldraw[fill=black] (0, 3.5) circle (0.04);
	
	\filldraw[fill=black] (3, 3.5) circle (0.04);
	\coordinate [label=right:{\footnotesize$B_1$}] (A) at (2.9, 3.5);
	
	\filldraw[fill=black] (0, 5.5) circle (0.04);
	\coordinate [label=left:{\footnotesize$C$}] (A) at (0.1, 5.6);
	
	\filldraw[fill=black] (3, 5.5) circle (0.04);
	\coordinate [label=right:{\footnotesize$C_1$}] (A) at (2.9, 5.6);
	
	\draw[->][](0.2,5.5) -- (2.8, 5.5);
	\draw[->][](0,5.4) -- (0, 3.7);
	\draw[->][](3, 5.4) -- (3, 3.7);
	\draw[-][snake=zigzag,segment amplitude=1pt, very thin](0.2,3.6) -- (2.8, 5.3);
	
	\coordinate [label=left:{\footnotesize$B$}] (A) at (7.1, 3.5);
	\filldraw[fill=black] (7, 3.5) circle (0.04);
	
	\filldraw[fill=black] (10, 3.5) circle (0.04);
	\coordinate [label=right:{\footnotesize$B_1$}] (A) at (9.9, 3.5);
	
	\filldraw[fill=black] (7, 5.5) circle (0.04);
	\coordinate [label=left:{\footnotesize$C$}] (A) at (7.1, 5.6);
	
	\filldraw[fill=black] (10, 5.5) circle (0.04);
	\coordinate [label=right:{\footnotesize$C_1$}] (A) at (9.9, 5.6);
	
	\draw[-][snake=zigzag,segment amplitude=1pt, very thin](7.2,5.5) -- (9.8, 5.5);
	\draw[->][](7.2,5.4) -- (9.8, 3.7);
	\draw[->][](7,5.4) -- (7, 3.7);
	\draw[->][](10, 5.4) -- (10, 3.7);
	
	\coordinate [label=left:{\footnotesize$B$}] (A) at (14.1, 3.5);
	\filldraw[fill=black] (14, 3.5) circle (0.04);
	
	\filldraw[fill=black] (17, 3.5) circle (0.04);
	\coordinate [label=right:{\footnotesize$B_1$}] (A) at (16.9, 3.5);
	
	\filldraw[fill=black] (14, 5.5) circle (0.04);
	\coordinate [label=left:{\footnotesize$C$}] (A) at (14.1, 5.6);
	
	\filldraw[fill=black] (17, 5.5) circle (0.04);
	\coordinate [label=right:{\footnotesize$C_1$}] (A) at (16.9, 5.6);
	
	\draw[-][snake=zigzag,segment amplitude=1pt, very thin](14.2,5.5) -- (16.8, 5.5);
	\draw[<-][](14.2,5.4) -- (16.8, 3.7);
	\draw[<-][](14,5.4) -- (14, 3.7);
	\draw[<-][](17, 5.4) -- (17, 3.7);
	
	\coordinate [label=left:{\footnotesize subcase 2.2.}] (A) at (-3.1, 2.5);
	\coordinate [label=left:{\footnotesize subcase 2.3.}] (A) at (3.9, 2.5);
	\coordinate [label=left:{\footnotesize subcase 4.2.}] (A) at (11, 2.5);
	\coordinate [label=left:{\footnotesize subcase 5.2.}] (A) at (18, 2.5);

	\end{tikzpicture}
	\caption{\footnotesize\textsf{Cases from group B in the proof of Claim \ref{claim-dla-dwoch}. Cliques $C$ and $C_1$ are in the same order-relationship with $B_1$.}}\label{fig:caseB}
\end{figure}

\medskip\noindent{\sc Subcase 2.2.}  $C\mn_\fA C_1$, $C\mn_\fA B$ and $B\mn_{\fA} C_1$  (see Figure \ref{sub22}).\\
Unlike other cases, we start from the structure $\exba$ defined by  adding $D$ as a copy of $B_1$, so  (recall)  $\langle D,B_1\rangle_{\exba} \equiv_{f,\id} \langle B,B_1\rangle_{\fA}$ and
$\langle D,A\setminus B_1\rangle_{\exba} \equiv_{f_1,\id} \langle B_1,A\setminus B_1\rangle_{\fA}$. 

In this case (similarly to subcase 2.1.) we have
\begin{itemize}
	\item $B\mn_{\fA} B_1$,
	\item $C\mn_{\fA} B_1$,
	\item $C\mn_{\exba} D$, $B\mn_{\exba} D$ and $C_1\mn_{\exba} D$,
	\item $D\mn_{\exba} B_1$.
	\item $\cont{C,D}{\exba}\equiv_{\id,f_1} \cont{C,B_1}{\fA}$. 
\end{itemize}
We modify $\exba$ to obtain a structure $\fA_1$ in which elements of $C\cup E$ will have their $\gamma_i$-witnesses also in $D$. This is done in two steps. First we modify the connection type $\cont{C,D}{\exba}$ to ensure required witnesses for elements in $C$ as follows. 
\begin{figure}[htb]
	\begin{tikzpicture}[scale=0.8]
	
	\draw (-1.4, 3.3) ellipse (1.5 and .5);
	\coordinate [label=left:{\footnotesize$B$}] (A) at (-2.7, 3.5);
	\filldraw[fill=black] (-1.4, 3.5) circle (0.03);
	\coordinate [label=right:{\footnotesize$f(d)$}] (A) at (-1.4, 3.3);
	\draw[->][style=dashed](-2, 5.9) .. controls (-1,5) and (-2.5,4.2) .. (-1.5,3.5);
	\filldraw[fill=black] (-2, 6) circle (0.03);
	\coordinate [label=left:{\footnotesize$E\ni e$}] (A) at (-2, 6);
	\draw[->][style=double](-2,6.1) .. controls (-1,6) and (0,8) .. (1.4, 7.4);

	\filldraw[fill=black] (4, 3.5) circle (0.04);
	\coordinate [label=right:{\footnotesize$B_1$}] (A) at (4, 3.5);
	
	\filldraw[fill=black] (0, 5.5) circle (0.04);
	\coordinate [label=left:{\footnotesize$C$}] (A) at (0, 5.6);
	
	\filldraw[fill=black] (4, 5.5) circle (0.04);
	\coordinate [label=right:{\footnotesize$C_1$}] (A) at (4, 5.6);
	
	\draw (1.5,7.5) ellipse (1.5 and .5);
	\coordinate [label=right:{\footnotesize$D$}] (A) at (2.8, 7.8);
	\filldraw[fill=black] (1.5, 7.4) circle (0.03);
	\coordinate [label=right:{\footnotesize$d$}] (A) at (1.5, 7.5);
	
	\draw[->][style=thick](0.2,5.5) -- (3.8, 5.5);
	\draw[->][style=thick](0.2,5.4) -- (3.8, 3.6);
	\draw[->][style=thick][style=thick]
	(0.2,3.5) -- (3.9, 3.5);
	\draw[->][style=thick](0,5.4) -- (0, 3.7);
	\draw[->][style=thick](4, 5.4) -- (4, 3.7);
	\draw[->][style=double, thick](0.1,5.6) -- (1.7, 6.8);
	\draw[->][style=thick]
	(3.9, 5.6) -- (2.1, 6.9);
	
	\draw[->][style=thick](0.1,3.7) -- (1.9, 6.8);
	\draw[->][style=thick](0.2,3.6) -- (3.8, 5.3);
	\draw[->][style=thick](2,6.8) -- (3.9, 3.7);
	
	\end{tikzpicture}
	\caption{\footnotesize\textsf{Subcase 2.2 in the proof of Claim \ref{claim-dla-dwoch}:  $C\mna C_1$, $C\mna B$, $B\mna C_1$ and $D$ is a fresh copy of the clique $B_1$ given by Claim \ref{claim-single-dupl}. Then $C\mnb D$, $B\mnb D$, $C_1\mnb D$, $D\mnb B_1$ and $\cont{C,D}{\exba}\equiv_{\id,f_1} \cont{C,B_1}{\fA}$. Single arrows depict types in $\fA_1$  that remain as in $\exba$. The connection type $\cont{C,D}{\exba}$ will be transferred from $\cont{C,B}{\exba}$ in Step 2.2.a (this is depicted by a double arrow).  If $f(d)$ is a $\gamma_i$-witness of $e\in E$ then in  $\fA_1$ the  type  $tp^{\exba}[e,d]$ is replaced by $tp^{\exba}[e,f(d)]$ in Step~2.2.b (for every $e\in E$,  depicted by dashed and double arrows, respectively). }}\label{sub22}
\end{figure}

\medskip\noindent
{\em Step 2.2.a. } 
\medskip\noindent Define 
$\cont{C,D}{\fA_1}:=_{\id,f} \cont{C,B}{\exb}$. 

\medskip Note that since $D$ is a copy of $B_1$ which is isomorphic to $B$, after  the modification of $tp^{\fA_1}[C,D]$ no incompatibility occurs. Additionally, we have ensured that every element of $C$ that has a $\gamma_i$-witness in $B$ is attributed a $\gamma_i$-witness in $D$,
thus ensuring condition {\em(ii)} for $C$ of our Claim:

for every $a\in C:  \quad$ if \quad ${W}_i^{\fA}(a)\cap B\neq \emptyset \mbox{ \quad then \quad}
{W}_i^{\fA_1}(a)\cap D\neq \emptyset.$

\noindent Also, since only a 
finite number of 2-types are changed and $\exba$ is witness-saturated, every element of $C\cup D$ has its witness somewhere outside $C\cup D$, as it had before. 
\medskip

It remains to guarantee that every  element  $e\in E$ has its  $\gamma_i$-witness in $D$ in case when it had a $\gamma_i$-witness in $B$. This is done by a modification of the connection type  
$\cont{E,B}{\exba}$ as follows.

\medskip\noindent
{\em Step 2.2.b.} 
\medskip\noindent For every $e\in E$, for every $d\in
D$, if $f(d)\in {W}_i^{\fA}(e)$, then 

\hspace{2cm}replace $tp^{\exba}[e,d]$
by $tp^{\exba}[e,f(d)]$.

\medskip After performing Step 2.2.b we have ensured that 
condition {\em(ii)} of our Claim holds for $E$.
Let us note that if $e\in E$ and  ${W}_i^{\fA}(e)\cap B\neq \emptyset$ then  $e\mn_{\exba} B_1$ and so by construction  we have also $e\mn_{\exba} D$ (cf.~condition (5) of Claim~\ref{claim-single-dupl}). So, in the above step no incompatibility occurs.
Moreover, observe that only a finite number of 2-types are changed in Step 2.2.b, so even if an element $d\in D$ lost a witness in $E$, $d$ still has another witness outside $E$ (since $\exba$ was witness-saturated). Hence, the structure obtained satisfies all conditions of our claim.

\medskip	
\noindent {\bf C}. {\em Remaining cases.}\\
Here we have two subcases: 4.3 and 5.3 (cf.~Figure \ref{fig:caseC}) that again can be handled in the same way. The construction proceeds by modifying the structure $\exb$. In these two cases the modification attributing required witnesses for elements from $C_1$  replaces some 2-types from $\BBBn$ by 2-types from $\BBBr$. Hence an additional step is needed to ensure that the interpretation of $T$ is transitive in the newly constructed structure. 
We give the details for Subcase 4.3, the argument for Subcase~5.3 is symmetric.

\begin{figure}[htb]
	\begin{tikzpicture}[scale=0.45]	
	\coordinate [label=left:{\footnotesize$B$}] (A) at (0.1, 3.5);
	\filldraw[fill=black] (0, 3.5) circle (0.04);
	
	\filldraw[fill=black] (3, 3.5) circle (0.04);
	\coordinate [label=right:{\footnotesize$B_1$}] (A) at (2.9, 3.5);
	
	\filldraw[fill=black] (0, 5.5) circle (0.04);
	\coordinate [label=left:{\footnotesize$C$}] (A) at (0.1, 5.6);
	
	\filldraw[fill=black] (3, 5.5) circle (0.04);
	\coordinate [label=right:{\footnotesize$C_1$}] (A) at (2.9, 5.6);
	
	\draw[-][snake=zigzag,segment amplitude=1pt, very thin](0.2,5.5) -- (2.8, 5.5);
	\draw[-][snake=zigzag,segment amplitude=1pt, very thin](0.2,5.4) -- (2.8, 3.7);
	\draw[->][](0,5.4) -- (0, 3.7);
	\draw[->][](3, 5.4) -- (3, 3.7);
	\draw[-][snake=zigzag,segment amplitude=1pt, very thin](0.2,3.6) -- (2.8, 5.3);
	
	\coordinate [label=left:{\footnotesize$B$}] (A) at (7.1, 3.5);
	\filldraw[fill=black] (7, 3.5) circle (0.04);
	
	\filldraw[fill=black] (10, 3.5) circle (0.04);
	\coordinate [label=right:{\footnotesize$B_1$}] (A) at (9.9, 3.5);
	
	\filldraw[fill=black] (7, 5.5) circle (0.04);
	\coordinate [label=left:{\footnotesize$C$}] (A) at (7.1, 5.6);
	
	\filldraw[fill=black] (10, 5.5) circle (0.04);
	\coordinate [label=right:{\footnotesize$C_1$}] (A) at (9.9, 5.6);
	
	\draw[-][snake=zigzag,segment amplitude=1pt, very thin](7.2,5.5) -- (9.8, 5.5);
	\draw[-][snake=zigzag,segment amplitude=1pt, very thin](7.2,5.4) -- (9.8, 3.7);
	\draw[<-][](7,5.4) -- (7, 3.7);
	\draw[<-][](10, 5.4) -- (10, 3.7);
	\draw[-][snake=zigzag,segment amplitude=1pt, very thin](7.2,3.6) -- (9.8, 5.3);
	
	\coordinate [label=left:{\footnotesize  subcase 4.3.}] (A) at (3.8, 2.5);
	\coordinate [label=left:{\footnotesize subcase 5.3.}] (A) at (10.9, 2.5);
	
	\end{tikzpicture}
	\caption{\footnotesize\textsf{Case C in the proof of Claim \ref{claim-dla-dwoch} consists of two subcases: 4.3 and 5.3.}}\label{fig:caseC}
\end{figure}

\medskip\noindent{\sc Subcase 4.3.} $C\simb C_1$, $C\mnb B$, $C_1\simb B$ 
and  $C\simb B_1$  (cf.~Figure~\ref{piccase42a}).

\noindent In this case (recalling our notation $X\lesssim Y$ iff $X<Y$ or $X \sim Y$) we have
\begin{enumerate}[a.]
	\item\label{waapp} $B\simb B_1$, (otherwise, if $B_1\mnb B$, then by transitivity $C_1\mnb B$; a contradiction with $C_1\simb B$, the same if $B\mnb B_1$),
	\item\label{wbapp} $C_1\simb D$ (by (5) in Claim~\ref{claim-single-dupl}),
	\item\label{wcapp} $C\mnb D$, by (*), \nb{L wyciete 'and'}
	\item\label{wcaapp} $B\simb D$, by construction. 
	
	\medskip
	\hspace{-8mm}
	\noindent Moreover, if  $V$ and $V_1$ are cliques in $\exb$ such that   $D \mnb V$ and
	$V_1\mnb C_1$ then we have also: 
	
	\medskip
	\item\label{weapp} $C\mnb V$ (since by c. $C\mnb D$),
	\item\label{wfapp} $V_1\lsimb B$  (otherwise, if $B \mnb V_1$ then by transitivity we have $B \mnb C_1$, a contradiction with $C_1\simb B$),
	\item\label{wdapp} $C_1 \lsimb V$ (otherwise  $D\mnb V \mnb C_1$, a contradiction with \ref{wbapp}.),
	\item\label{wgapp} $V_1\lsimb D$  (by construction of $\exb$ and by \ref{wfapp}.), 
	\item\label{whapp} $V_1\lsimb V$  (otherwise, if $V\mnb V_1$, then by transitivity we get $D\mnb V_1$; a contradiction with \ref{wgapp}.).
\end{enumerate}

\begin{figure}[thb]
	\begin{tikzpicture}[scale=0.8]
	
	\filldraw[fill=black] (0, 3.5) circle (0.06);
	\coordinate [label=left:{\footnotesize$B$}] (A) at (0, 3.1);
	
	\filldraw[fill=black] (4, 3.5) circle (0.06);
	\coordinate [label=right:{\footnotesize$B_1$}] (A) at (4, 3.1);
	
	\filldraw[fill=black] (0, 5.5) circle (0.06);
	\coordinate [label=left:{\footnotesize$C$}] (A) at (0, 5.1);
	
	\filldraw[fill=black] (4, 5.5) circle (0.06);
	\coordinate [label=right:{\footnotesize$C_1$}] (A) at (4, 5.1);
	
	\filldraw[fill=black] (2, 7) circle (0.09);
	\coordinate [label=right:{\footnotesize$D$}] (A) at (2, 7.4);
	
	\draw[-][snake=zigzag,segment amplitude=1pt, very thin](0.3,3.7) -- (3.8, 5.3);
	\draw[-][snake=zigzag,segment amplitude=1pt, very thin](0.1,5.5) -- (3.8, 5.5);
	\draw[-][snake=zigzag,segment amplitude=0.7pt, very thin](0.2,5.4) -- (3.8, 3.7);
	\draw[-][snake=zigzag,segment amplitude=0.7pt, very thin](0.3,3.5) -- (3.8, 3.5);
	\draw[-][snake=zigzag,segment amplitude=0.7pt, very thin](0.15,3.7) --  (2, 6.8);
	\draw[->][thick](0,5.4) -- (0, 3.7);
	\draw[->][thick](4, 5.4) -- (4, 3.7);
	\draw[->][thick](0.1,5.6) -- (1.8, 6.9);
	\draw[-][snake=zigzag,segment amplitude=0.7pt, very thin](3.9,5.7) -- (2.2, 6.8);

	\filldraw[fill=black] (-3, 5.5) circle (0.06);
	\coordinate [label=left:{\footnotesize$V$}] (A) at (-3, 5.1);
	\draw[<-][thick](-2.8, 5.9).. controls(-2, 7) .. (1.7, 7);
	\draw[<-][dotted, thick](-2.8, 5.7).. controls(-0.5, 6.5) .. (3.85, 5.6);
	\draw[<-][thick](-2.8, 5.5) -- (-0.1, 5.5);
	\draw[<-][thick](-2.8, 5.3) -- (-0.1, 3.6);
	
	\filldraw[fill=black] (7, 5.5) circle (0.06);
	\coordinate [label=right:{\footnotesize$V_1$}] (A) at (7, 5.1);
	\draw[->][thick](6.9, 5.4) -- (4.2, 3.6);
	\draw[->][thick](6.9, 5.5) -- (4.2, 5.5);
	\draw[->][dotted, thick](6.9, 5.7) .. controls (5,8) and (-2.5,9) .. (-3,5.8);
	\draw[->][dotted, thick](6.6, 5.7) .. controls (5.1,6.6)  .. (2.2,7);
	\draw[->][dotted, thick](6.6, 5.4) .. controls (3.4,4.1)  .. (0.3, 3.6);
	\end{tikzpicture}
	\caption{\footnotesize\textsf{Subcase 4.3 in the proof of Claim \ref{claim-dla-dwoch}: $C\simb C_1$, $C\mnb B$, $C_1\simb B$, $C\simb B_1$ and  $D$ is a fresh copy of the clique $B$ given by Claim \ref{claim-single-dupl}, so $C_1\simb D$. The cliques $V$ and $V_1$ satisfy   $D \mnb V$ and $V_1\mnb C_1$. Single arrows depict types from $\BBBr$ that will  remain in $\fA_1$ as in
			$\exb$. Dotted arrows mean that the according types are in $\BBBr\cup\BBBn$ (i.e. $V_1\lsimb B$, $C\lsimb V$, and so on). Snake lines connect incomparable cliques. The connection type	$\cont{C_1,D}{\fA_1}$ will be transferred from	 $\cont{C_1,B_1}{\exb}$	  in Step  4.3.a.}}\label{piccase42a}
\end{figure}

To ensure that elements of $C_1$ have their $\gamma_i$-witnesses in $D$ as they had them in $B_1$ 
we start by modifying the connection type $\cont{C_1,D}{\exb}$ by transferring the connection type from $\cont{C_1,B_1}{\exb}$ as follows.

\medskip\noindent{\em Step 4.3.a.} 
\medskip\noindent  
Define $\cont{C_1,D}{\fA_1} :=_{\id,f_1}\cont{C_1,B_1}{\exb}$ (see Figure \ref{piccase42a}). 

\medskip Note that since $D$ is a copy of $B$ which is isomorphic to $B_1$, the above modification of $\cont{C_1,D}{\exb}$ is well defined. Additionally we have ensured, that every element of $C_1$  has its $\gamma_i$-witness in $D$ as before in $B_1$ (thus ensuring condition {\em(iii)} of Claim \ref{claim-dla-dwoch}).  Also, since only $\BBBn$-types are changed, no witness is stolen.  
However, since $C_1\simw_\exb D$ and after Step 4.3.a. $C_1\mnaa D$, it is not guaranteed that $T$ is transitive in $\fA_1$. 
To make $T$ transitive in $\fA_1$ we extend (in a minimal way) the partial order on the set of cliques in $\exb$ so that it contains the pair $(C_1,D)$, as described in Proposition~\ref{prop:partial-extend}.  
More precisely, the goal  of the next step is to ensure that 

\begin{verse}
	for every clique $V$ such that $D\mnb V$ or $V=D$, and  
	
	for every clique  $V_1$ such that $V_1\mnb C_1$ or $V_1= C_1$
	
	\noindent we will have 
	$V_1\mnaa V$ (see Figure \ref{piccase43a}).
\end{verse}

\begin{figure}[htb]
	\begin{tikzpicture}[scale=0.8]
	
	\filldraw[fill=black] (0, 3.5) circle (0.06);
	\coordinate [label=left:{\footnotesize$B$}] (A) at (0, 3.1);
	
	\filldraw[fill=black] (4, 3.5) circle (0.06);
	\coordinate [label=right:{\footnotesize$B_1$}] (A) at (4, 3.1);
	
	\filldraw[fill=black] (0, 5.5) circle (0.06);
	\coordinate [label=left:{\footnotesize$C$}] (A) at (0, 5.1);
	
	\filldraw[fill=black] (4, 5.5) circle (0.06);
	\coordinate [label=right:{\footnotesize$C_1$}] (A) at (4, 5.1);
	
	\filldraw[fill=black] (2, 7) circle (0.09);
	\coordinate [label=right:{\footnotesize$D$}] (A) at (2, 7.4);
	
	\draw[-][snake=zigzag,segment amplitude=0.7pt, very thin](0.2,3.6) -- (3.8, 5.4);
	\draw[-][snake=zigzag,segment amplitude=0.7pt, very thin](0.2,5.4) -- (3.8, 3.7);
	\draw[-][snake=zigzag,segment amplitude=0.7pt, very thin](0.3,3.5) -- (3.8, 3.5);
	\draw[-][snake=zigzag,segment amplitude=0.7pt, very thin](0.1,5.5) -- (3.8, 5.5);
	\draw[->][thick](0,5.4) -- (0, 3.7);
	\draw[->][thick](4, 5.4) -- (4, 3.7);
	\draw[->][thick](0.1,5.6) -- (1.8, 6.9);
	\draw[->][double, thick](3.9,5.7) -- (2.1, 6.8);

	\filldraw[fill=black] (-3, 5.5) circle (0.06);
	\coordinate [label=left:{\footnotesize$V$}] (A) at (-3, 5.1);
	\draw[<-][thick](-2.7, 6).. controls(-2, 7) .. (1.7, 7);
	\draw[<-][double, thick](-2.8, 5.7).. controls(-0.5, 6.5) .. (3.85, 5.6);
	\draw[<-][thick](-2.8, 5.5) -- (-0.1, 5.5);
	\draw[<-][thick](-2.8, 5.3) -- (-0.1, 3.6);
	
	\filldraw[fill=black] (1.6, 1.5) circle (0.04);
	\coordinate [label=right:{\footnotesize$v_1'$}] (A) at (1.5, 1.3);
	\draw[->](1.7, 1.9) -- (1.9, 6.7);
	\draw[->](1.4, 1.8) -- (0.1, 3.3);
	
	\draw[->](1.4, 1.7) .. controls (-1,2) and (-1.5,4) .. (-2.9, 5.1);
	\draw[->](1.55, 1.8) .. controls (0.5,5) and (0.5,5) .. (0.1, 5.3);
	
	\filldraw[fill=black] (7, 5.5) circle (0.06);
	\coordinate [label=right:{\footnotesize$V_1$}] (A) at (7, 5.1);
	\draw[->][thick](6.9, 5.4) -- (4.2, 3.6);
	\draw[->][thick](6.9, 5.5) -- (4.2, 5.5);
	\draw[->][double, thick](6.9, 5.7) .. controls (5,8) and (-2.5,9) .. (-3,5.8);
	\draw[->][double, thick](6.6, 5.7) .. controls (5.1,6.6)  .. (2.2,7);

	\end{tikzpicture}
	\caption{\footnotesize\textsf{Subcase 4.3 in the proof of Claim \ref{claim-dla-dwoch}: the intended model $\fA_1$ after Step 4.3.b. $\fA_1$ is obtained modifying connection types $\cont{C_1,V}{\exb}$, $\cont{V_1,D}{\exb}$ and $\cont{V_1,V}{\exb}$ (this is depicted by double arrows). Single arrows depict types from $\BBBr$ that remain in $\fA_1$ as in $\exb$. }}\label{piccase43a}
\end{figure}

\medskip\noindent{\em Step 4.3.b.}
For every clique $V$  such that $D \mnb V$ (see Figure \ref{piccase43a})
\begin{enumerate}
	\item if $\neg(C_1\mnb V)$, then $\cont{C_1,V}{\fA_1}:=_{f_C,\id}\cont{C,V}{\exb}$, where $f_C$ is any isomorphism from $C_1$ to $C$. 
	\item For every clique  $V_1$ such that $V_1\mnb C_1$
	\begin{enumerate}
		\item
		if $\neg(V_1 \mnb D)$ then 
		for every $v_1\in V_1$\\ find $v_1'\in A_{+B}$ such that $v_1' \mnb C$ and $tp^{\exb}[v_1']=tp^{\exb}[v_1]$, \\
		for every $d\in D$ define $tp^{\fA_1}[v_1,d]$ as $tp^{\exb}[v_1',d]\in\BBBr$. 
		\item 
		if $\neg(V_1 \mnb V)$ then for every $v_1\in V_1$,\\ 
		find $v_1'\in A_{+B}$ such that $v_1'\mnb C$ and $tp^{\exb}[v_1']=tp^{\exb}[v_1]$, \\
		for every $v\in V$ 
		define $tp^{\fA_1}[v_1,v]$ as $tp^{\exb}[v_1',v]\in\BBBr$. 

	\end{enumerate} 
	
\end{enumerate} 
\medskip
Observe that if $\neg(C_1 \mnb V)$ then by \ref{wdapp}. $V \simb C_1$ and by \ref{weapp}., $C \mnb V$.  This implies that after performing line (1) in Step~4.3.b.~$C_1\mnaa V$ as desired.

Now, we argue that the required elements $v_1'$ in lines (a) and (b) can always be found. 
Recall that it is assumed  $\spl^{\exb}(C_1)=\spl^{\exb}(C)$ that in particular implies $\In^{\exb}(C_1)=\In^{\exb}(C)$. Now, $v_1\in V_1$ implies $tp^\exb(v_1)\in \In^\exb(C_1)= \In^\exb(C)$. Hence, the required $v_1'$ in line (a) can be found in some clique $V_1'\mnb C$. Then by \ref{wcapp}.~we have also  $V_1'\mnb D$, so $v_1'\mnb D$ (also  $v_1'\mnb V$) and each of the types $tp^\exb[v_1',d]$ belongs to $\BBBr$, as indicated. Similar argument applies to line (b).

Since $T$ was transitive in $\exb$, then after performing the above step, the relation $\mnaa$ is a strict partial order on the set of cliques of $\fA_1$, and by Proposition~\ref{l:partialorder}, $T$ is transitive in $\fA_1$.

Finally, observe that in $\exb$ only types from $\BBBn$ were modified   (cf. \ref{wdapp}.--\ref{whapp}.), 
so all  elements of $\fA_1$ have their witnesses as before in $\exb$ (recall, incomparable witnesses in $\fodeg$ are not allowed). Moreover, whenever we
replaced a 2-type $\beta\in \BBBn$, realized by some pair $(a,b)$, by a 2-type $\beta'\in\BBBr$, then $\beta$ and $\beta'$ agreed with the 1-types contained, namely $\beta \obciety x =\beta' \obciety x$ and $\beta \obciety y = \beta' \obciety y$, so the construction preserved cliques and witnesses within cliques. 
All modified 2-types were replaced by 2-types realized in $\exb$, so we have ensured that the structure $\fA_1$ fulfills conditions {\em(a)-(d)} of Proposition  \ref{claim-iff}, i.e.~$\fA_1\models \Psi$. 
Finally, $\spl^{\fA_1}(D)=\spl^{\fA_1}(B)= \spl^{\exb}(B)$, as in line (a) every 1-type $tp^\exb[v_1]\in \In^\exb(C_1)\subseteq \In^\exb(B_1)=\In^\exb(B)$, so $\In^{\fA_1}(D)=\In^{\exb}(B)$.

\end{proof}

Now we are ready to show the announced Corollary \ref{wniosek-dla-wielu}. It says that when several ele\-ments $a_1,
a_2,\ldots,a_p$ of a model $\fA$ have $\gamma_i$-witnesses in
several distinguished cliques that realize the same splice, one
can extend $\fA$ by a single clique $D$ (realizing the
same splice) in which $a_1, a_2,\ldots,a_p$ have their
$\gamma_i$-witnesses.  In the proof we iteratively apply Claim \ref{claim-dla-dwoch} for the cliques $C=Cl^{\fA}(a_1)$ and $C_1=Cl^{\fA}(a_i)$, where $i=2,\ldots,p$.

\begin{corollary}\label{wniosek-dla-wielu}
Assume $\fA$ is countable witness-saturated, $\gamma_i\in \Psi$, 
$\mathcal C=\{V_1,\ldots,V_p\}, \mathcal B=\{U_1,\ldots,U_p\}\subseteq Cl^\fA\setminus {\mathbb K}(\fA)$ and for every $k$ ($1\leq k \leq p$):

\hspace{1cm}$\spl(U_1)=\spl(U_k)$, $\spl(V_1)=\spl(V_k)$  and 

\hspace{1cm}there is $a\in V_k$ such that $W_i(a)\cap U_k\neq \emptyset$. 

\noindent Then, there
is an  extension $\fA'$ of $\fA$ by at most one clique $D\subseteq A'$
such that
\begin{enumerate}[(i)]
\item $\fA'\models \Psi$ and $\fA'$ is  witness-saturated,
\item for every $k$ ($1 \leq k \leq p$), for every $a\in V_k$, if $W_i^\fA(a)\cap U_k\neq \emptyset$, then ${W}_i^{\fA'}(a)\cap D\neq\emptyset,$ 
\item $\spl^{\fA'}(D)= \spl^\fA(U_1). $
\end{enumerate}

\end{corollary}

\begin{proof} 
Iteratively applying Claim~\ref{claim-dla-dwoch} we construct a sequence of cliques,  $D^{(1)},D^{(2)},\ldots,D^{(p)}$ and a sequence of models $\fA^{(1)},\fA^{(2)},\ldots,\fA^{(p)}$ such that $\fA^{(k)}={\fA_1^{(k-1)}}$ (cf. Claim 
\ref{claim-dla-dwoch}) and all the elements of the cliques $V_1, \ldots, V_{k}$ have their
$\gamma_i$-witnesses in the clique $D^{(k)}$, $k=1,\ldots p$. The statement of the corollary is obtained by setting $\fA' := \fA^{(p)}$ and $D := D^{(p)}$.

The interesting case is when $p>2$. W.l.o.g.~assume that the clique $V_1$ is a minimal element in $\mathcal C$ under the clique-order on $\fA$ and set  $C:=V_1$. The desired model $\fA'$ will be constructed in $p$ steps. 

\medskip\noindent
{\em Step 1.} Define $D^{(1)}=U_1$ and $\fA_{}^{(1)}=\fA$.

\medskip
Trivially, for every $a\in V_1$, if $W_i^\fA(a)\cap U_1\neq \emptyset$ then ${W}_i^{\fA^{(1)}}(a)\cap D^{(1)}\neq\emptyset$ and $\spl^{\fA^{(1)}}(D^{(1)})= \spl^\fA(U_1)$.

In the next steps, using Claim 
\ref{claim-dla-dwoch},  we will construct  a sequence of models  $\fA^{(k)}={\fA_1^{(k-1)}}$  such that all the elements of $\bigcup_{j=1}^k V_j$ will have their
$\gamma_i$-witnesses in  $D^{(k)}$.
The following  invariant  is maintained in the process: 
\begin{itemize}
	\item[-] $\fA^{(k)}\models \Psi$ and $\fA^{(k)}$ is  witness-saturated,
	\item[-] ${W}_i^{\fA^{(k)}}(a)\cap D^{(k)}\neq\emptyset,$ for $a\in \bigcup_{j=1}^{k}V_j$ such that $ W_i^\fA(a)\cap U_j\neq \emptyset$, 
	\item[-] $\spl^{\fA^{(k)}}(D^{(k)})= \spl^\fA(U_1). $
\end{itemize}

\medskip\noindent
{\em Step k+1} ($1\leq k \leq p-1$). 
Set $B:= D^{(k)}$, $B_1:= U_{k+1}$ and $E=\bigcup_{1\leq j\leq k} C_j$.
\begin{enumerate}
	\item If $B=B_1$, then all desired $\gamma_i$-witnesses are already provided by $B$. Define $D^{(k+1)}=B$ and $\fA^{(k+1)}=\fA^{(k)}$. No clique is added.
	\item Assume 
	 $B\neq B_1$. Then 
\begin{enumerate}
	\item set $C_1:= V_{k+1}$ (observe $C_1 \nmna C$, as $C$ was  minimal  in $\mathcal C$),
	\item apply Claim \ref{claim-dla-dwoch} for $\fA^{(k)}$, $B$, $B_1$, $C$, $C_1$ and $E$: \\
	set  $\fA^{(k+1)}:={\fA^{(k)}_1}$
	and $D^{(k+1)}:=D$. 
\end{enumerate}
\end{enumerate} 

\medskip\noindent
 Obviously, the invariant is maintained after performing Step~$k+1$. Observe that it is possible that during the  construction no clique is added ($D^{(1)}=D^{(2)}=\ldots=D^{(p)}=U_1$) and then $\fA'=\fA$.

\medskip Finally, let $D=D^{(p)}$ and  $\fA'=\fA^{(p)}\obciety(A\cup D^{(p)})$. It is obvious that $D$ provides all $\gamma_i$-witnesses, as required in condition {\em(ii)} and additionally, the segment $\bigcup^{p-1}_{k=2}D^{(k)}$ is redundant in $\fA^{(p)}$.
By Proposition \ref{prop-redundant-first} and Claim \ref{claim-dla-dwoch} conditions {\em(i)} and {\em(iii)} of our claim  also hold.
\end{proof}

\subsection{Construction of a narrow model}\label{sec:proof-of-narrow-lemma}

We first prove a generalization of Corollary
\ref{wniosek-dla-wielu}. It says, roughly speaking, that if
$\fA\models\Psi$ and $F$ is a finite subset of $A$, then it is
possible to extend $\fA$ by a bounded-size segment such that
$\gamma_i$-witnesses for all elements of $F$ can be found in this segment, for every $i$ ($1\leq i \leq m$). The new segment is of  size bounded by $\mathbb M$ defined by the formula \eqref{def:M} and, in particular, independent of the size of $F$.

\begin{lemma}[{\rm Witness compression}]
\label{lem:wit-compr} Assume $\fA$ is a countable
witness-saturated model of $\Psi$ and $F\subseteq
A\setminus{K}(\fA)$ is finite. There is a witness-saturated
extension $\fA^*$ of $\fA$ such that the universe $A'=A \dot{\cup}S$, $S$ is a segment in $\fA'$ and:
\begin{enumerate}
\item $\fA'\models \Psi$, 
\item $|S|\leq {\mathbb M},$
\item for every conjunct $\gamma_i$ of $\Psi,$  for
every $a\in F$, if ${W}_i^{\fA}(a)\setminus {K}(\fA)\neq
\emptyset$, then  ${W}_i^{\fA^*}(a)\cap S\neq \emptyset$.
\end{enumerate}
\end{lemma}

\begin{proof}
Let us analyse the situation. Assume $F$ is as above, $a\in F$ and $b\in {W}_i^{\fA}(a)$. Elements $a$ and $b$ realize some splices, say $\spl(Cl^\fA(a))= Y$ and $\spl(Cl^\fA(b))= X$. The idea is to distinguish a set $F_i^{Y,X}\subseteq F$ of elements that realize the splice $Y$ and  have their $\gamma_i$-witnesses in cliques realizing the splice $X$. Then we can apply Corollary \ref{wniosek-dla-wielu} to get a single clique, say $D_i^{Y,X}$, in which all the elements of $F_i^{Y,X}$ have their $\gamma_i$-witnesses. The above procedure we repeat for every $i$ ($1\leq i\leq m)$ and for every $Y,X\in Sp^\fA$. At the end we get the segment $S$ as a union of (pairwise disjoint) cliques $D_i^{Y,X}$.

First, for every $i$ ($1\leq i\leq m)$ and for every $a \in F$ 
denote by $\bar{\gamma_i}(a)$ an arbitrarily chosen
element $b\in {W}_i^{\fA}(a)\setminus K(\fA)$ (the value of a Skolem function
for the existential quantifier in $\gamma_i$). For every $Y,X\in Sp^\fA$ set:
\begin{itemize}
	\item $F_i^{Y,X}:=\{a\in F: \spl(Cl^\fA(a))= Y \mbox{ and } \spl(Cl^\fA(\bar{\gamma_i}(a))= X\}$,
	\item $\mathcal C_i^{Y,X}:=\{Cl^\fA(a)): a\in F_i^{Y,X} \}$,
	\item $\mathcal B_i^{Y,X}:=\{Cl^\fA(\bar{\gamma_i}(a)): a\in F_i^{Y,X} \}$.
\end{itemize}  
To construct the required model $\fA^*$ and the segment $S$ we proceed as follows.
 
\medskip\noindent
For every $i$ ($1\leq i\leq m)$ for every  $Y,X\in Sp^\fA$:

apply Corollary~\ref{wniosek-dla-wielu} for the sets ${\mathcal C}=\mathcal C_i^{Y,X}$ and $\mathcal {B}=\mathcal B_i^{Y,X}$: 

replace $\fA$ by the structure $\fA'$ given by Corollary~\ref{wniosek-dla-wielu},

set $D_i^{Y,X}:=D$, if a clique  $D$ was added to $\fA$, otherwise set $D_i^{Y,X}:=\emptyset$.

\medskip\noindent
Denote the resulting structure by $\fA^*$. Condition $(i)$ of Corollary~\ref{wniosek-dla-wielu} implies that
$\fA^*\models \Psi$.  Let $S$ be the segment
consisting of elements of  the newly added cliques:
$$S\stackrel{def}{=} \bigcup_ {1\leq i\leq m}\,\,\,\bigcup_{{Y,X}\in
Sp^\fA}D_i^{Y,X}.$$
Obviously, $|S|\leq~m\cdot~ s^2\cdot~h~<~{\mathbb M}$, as 
the size of every clique is bounded by $h$ and the number of pairs of splices $Y$ and $X$ is bounded by $s^2$.

To show that condition (3) of our lemma holds, assume
$\gamma_i\in\Psi$, $a\in F$ and ${W}_i^{\fA}(a)\setminus
{K}(\fA)\neq \emptyset.$ Then there exists $b\in W_i^\fA(a)$ such
that $b\not\in K(\fA)$. So $a\in F_i^{Y,X}$, where $Y=\spl(Cl^\fA(a))$ and
$X=\spl(Cl^\fA(b))$.
Now, by condition~(ii) of Lemma \ref{wniosek-dla-wielu}, we obtain
${W}_i^{\fA^*}(a)\cap D_i^{Y,X}\neq\emptyset,$
and so,
${W}_i^{\fA^*}(a)\cap S\neq \emptyset$.
\end{proof}

Now we are ready to prove existence of narrow models. For convenience we recall the corresponding definition and the statement of Lemma~\ref{lem-narrow}.

\medskip\noindent{\bf Definition \ref{def-narrow}.}
{\em A model $\fA$ of $\Psi$ is  {\em narrow} if $A= K(\fA)$ or there is an infinite
	partition $P_A=\{S_0, S_1,\ldots\}$ of the universe $A$ such that $K(\fA) \subseteq S_{0}$ and for every $j\geq 0$:
	\begin{enumerate}
		\item $|S_j|\leq {\mathbb M}$,
		\item  for every $a\in \bigcup_{k=0}^j S_k$
		and for every $\gamma_i\in \Psi,$
		
		\hspace{0,5cm}if ${W}_i^{\fA}(a)\cap S_0=\emptyset$, then
		${W}_i^{\fA}(a)\cap S_{j+1}\neq \emptyset.$
	\end{enumerate}
}

\medskip\noindent{\bf Lemma \ref{lem-narrow}.} {\em\ 	Every 
	satisfiable \fodeg-sentence $\Psi$ has a narrow model.
}

\begin{proof}[Proof of Lemma \ref{lem-narrow}]
Assume $\fA$ is a witness-saturated model of $\Psi$
that exists by Lemma \ref{lemma-saturated}. In general our goal is to construct an infinite sequence of segments $K(\fA)\subset S_0, S_1,\ldots$ and a corresponding set of structures $\fA_0=\fA, \fA_1, \fA_2, \ldots$. Every $\fA_{j+1}$ is obtained using Lemma \ref{lem:wit-compr} for $F=\bigcup_{k=0}^{j} S_k$.
In particular setting $\fA'=(\bigcup_{k=0}^\infty\fA_k)\obciety\bigcup_{k=0}^\infty S_k$ proves the corollary.

If $A=K(\fA)$ then we are done. Otherwise, for $\gamma_i\in\Psi$
and $a\in {K}(\fA)$ denote by $\bar{\gamma_i}(a)$ an arbitrarily chosen
element $b\in {W}_i^{\fA}(a)$ (the value of a Skolem function
for the existential quantifier in $\gamma_i$). Define
$\fA_0=\fA$ and
$$S_0={K}(\fA)\cup
\bigcup_{1\leq i\leq m}\,\,\,\bigcup_{a\in{K}(\fA)}\,Cl^{\fA}(\bar{\gamma_i}(a)).$$
Note that in this case $K(\fA)\subsetneq  S_0$ and $|S_0|\leq~s\cdot~h~+~m\cdot~s\cdot~h~<~{\mathbb M}$.\\
 Iterating we define a sequence of structures $\fA_1, \fA_2, \ldots$ such that for each $j\geq 0$ we have
$\fA_{j+1}=\fA_j^*,$ where $\fA_j^*$ is the extension of $\fA_j$ by a segment $S_{j+1}$
given by Lemma \ref{lem:wit-compr} for $F=\bigcup_{k=0}^j S_k$. 
This means that each
$S_{j+1}$ extends $\fA_j$ to $\fA_{j+1}$,  $\fA_{j+1}\models \Psi$ and for every $\gamma_i\in\Psi,$ for
every $a\in \bigcup_{k=0}^j S_k$, if ${W}_i^{\fA}(a)\setminus {K}(\fA_j)\neq
\emptyset$, then ${W}_i^{\fA_{j+1}}(a)\cap S_{j+1}\neq \emptyset$.
Now, define $$\fA'=(\bigcup_{k=0}^\infty
\fA_k)\obciety\bigcup_{k=0}^\infty S_k.$$ By Proposition \ref{prop-redundant-first} and Lemma \ref{lem:wit-compr}, it is easy to
see that $\fA'$ is a narrow model of $\Psi$ with partition
$P_A=\{S_0, S_1,\ldots\}.$
\end{proof}

\section{Decidability of \fodeg{}}\label{sec:main}
As before we assume $\Psi$ is a normal-form \fodeg-formula, models of $\Psi$ have the exponential clique property and $\mathbb M$ defined by equation \eqref{def:M} is the bound on the size of segments in narrow models. 

\subsection{Regular models}

In this section we analyze properties of models of $\Psi$ on the
level of segments which consist of several cliques, and
constitute a partition $S_0,S_1,\ldots$ of the universe of a
model. Every segment $S_j$ has doubly exponential size
and is meant to contain all $\gamma_i$-witnesses for elements
from earlier segments $S_0,S_1,\ldots,S_{j-1}.$ On this level of abstraction
cliques and splices of a model become much less important.

\begin{definition}\label{color} 
	Assume  $\fA$ is a narrow model of $\Psi$ with a partition $P_A=\{S_0, S_1,\ldots\}$. We say that a connection type $\langle S_{j'},S_{k'}\rangle_{\fA}$ {\em is equivalent to}     $\langle S_j,S_k\rangle_\fA$, denoted 
	$$\langle S_{j'},S_{k'}\rangle \approx_{\fA}\langle S_j,S_k\rangle$$ iff $j'<k'$, $j<k$ and there exist isomorphisms $f_{S_j}$ and $f_{S_k}$ such that (cf.~Definition~\ref{def-conn-transfered})
	$$\langle S_{j'},S_{k'}\rangle_{\fA} \equiv_{f_{S_j},f_{S_k}} \langle S_j,S_k\rangle_{\fA}.$$
\end{definition}
If it is clear from the context we skip the subscript and simply write $\langle S_{j'},S_{k'}\rangle \approx\langle S_j,S_k\rangle$ instead of  $\approx_{\fA}$. It is obvious that  $\approx$ is an equivalence relation on the set $\{\langle S_j,S_k\rangle_{\fA}: 0\leq j<k\}$.

\begin{definition}\label{def:canon}
Assume  $\fA$ is a narrow model of $\Psi$ with a partition $P_A=\{S_0, S_1,\ldots\}$.  We say that
$\fA$ is  {\em regular} if $\fA$ is finite or (cf.~Figure \ref{fig-cann})  for every $k\in\N$, $0< k <\infty$:
$$\langle S_{k+1},S_{k+2}\rangle \approx \langle S_k,S_{k+1} \rangle\mbox{ and }\langle S_0,S_{k+1}\rangle \approx \langle S_{0},S_{k}\rangle.$$
\end{definition}

Note that in a regular model $\fA$  we have also $\langle S_{j},S_{k}\rangle \approx_{\fA}\langle S_{j'},S_{k'}\rangle$, for every $j,k,j',k'\geq 1$ with $j<k$ and $j'<k'$ (see Figure \ref{fig-cann}).

\begin{figure}[h]
	\begin{center}
		\xymatrix{
			*+[F=]{\,\,\,S_0\,\,\,}
			\ar@{-}[r]^{}
			\ar@{-}@/^1pc/[rr]^{}
			\ar@{-}@/^1,5pc/[rrr]^{}
			\ar@{-}@/^2,4pc/[rrrrr]^{}
			\ar@{-}@/^3,1pc/[rrrrrr]^{}
			& *+[F-]{\,\,S_1\,\,}
			\ar@{--}[r]^{}
			\ar@{--}@/_1pc/[rrr]^{}
			\ar@{--}@/_1,6pc/[rrrr]^{}
			\ar@{--}@/_2,4pc/[rrrrr]^{}
			& *+[F-]{S_{2}}
			\ar@{--}@/^0,5pc/[rr]^{}
			& *{\ldots}
			& *+[F-]{\,\,S_k\,\,}
			\ar@{--}[r]^{}
			\ar@{--}@/^1pc/[rr]^{}
			& *+[F-]{S_{k+1}}
			& *++{\ldots}\\
		}

 \caption{\footnotesize\textsf{Pattern of connection types in a regular model. Solid lines depict equivalent connection types: 
 		$\langle S_{0},S_{1}\rangle \approx \langle S_0,S_{2} \rangle$ $\approx$ $\langle S_{0},S_{3}\rangle \ldots$.
  		Dashed lines depict that $\langle S_{1},S_{2}\rangle\! \approx\! \langle S_1,S_{3} \rangle\! \approx \ldots \approx \langle S_{2},S_{3}\rangle\! \approx \! \langle S_2,S_{4} \rangle\ldots$.}}\label{fig-cann}
\end{center}
\end{figure}

\begin{lemma} Every  satisfiable \fodeg-sentence $\Psi$ has a regular
model. \label{twcanona}
\end{lemma}

\begin{proof}
	Let  $\fA$ be a narrow model of $\Psi$  given by Lemma~\ref{lem-narrow}. If $A=K(\fA)$ then $\fA$ is finite and there is nothing to prove. So, assume
	$X=\{S_0, S_1,\ldots\}$ is an infinite partition of $A$ given by Definition~\ref{def-narrow}.
	
	Observe that for every $k>0$, $S_k$ is redundant in $\fA$.  For,
	assume (cf. Definition \ref{def-refundant}) $b\in S_k,\, a\in
	A\setminus S_k$ and $b\in{W}_i^{\fA}(a).$ Assume $a\in S_l$ and
	take $j\in \N^+$ such that $j>\max\{k,l\}.$ By Definition
	\ref{def-narrow},  if ${W}_i^{\fA}(a)\cap
	S_0=\emptyset$, then ${W}_i^{\fA}(a)\cap S_{j+1}\neq \emptyset.$
	So, there is $c\in S_0\cup S_{j+1}$ such that
	$c\in{W}_i^{\fA}(a).$
	
	Similarly, for every infinite  $Z\subset
	\N^+$, the segment $\bigcup_{j\in\N^+\setminus Z} S_j$ is redundant in $ \fA$, and by Proposition \ref{prop-redundant-first}, $\fA \obciety \bigcup_{j\in Z\cup \{0\}} S_j \models \Psi$.

	The required regular model $\fA'$ of $\Psi$ is built as follows.
	Let  $[X]^2$ be the set of 2-element subsets of $X$.
	Now, define a colouring assigning to $[X]^2$:
	$$ Col({\{S_j,S_k\}})=[\langle S_{\min(j,k)},S_{\max(j,k)}\rangle]_{\approx}.$$
	
	So, the set $[X]^2$, is partitioned  into $c$ classes, where $c$ is the number of possible colours, which is finite. In this context, the infinite Ramsey theorem (cf.~e.g.~\cite{Die16}, Theorem 9.1.2) says
	that $X$ has an infinite monochromatic subset, say $Y$. At this point let us note that:
	$$\mbox{if }\langle S_{j},S_{k}\rangle \approx \langle S_k,S_{l} \rangle\mbox{ then }\fA\obciety S_j,\ \fA\obciety S_k \text{ and }\fA\obciety S_l \text{ are isomorphic}.$$
	So, all substructures of $\fA_{}$ induced by the segments from $Y$ are isomorphic.
	
	In a similar way, using  the basic pigeonhole principle, one can find an infinite set $Z\subseteq Y$ such that
	$\langle S_0,S_{k}\rangle \approx \langle S_{0},S_{l}\rangle,$ for every $S_k,S_l\in Z$. Now define $$\fA_{}'=\fA \obciety (S_0\cup \bigcup_{S\in Z} S).$$
	It follows from our preliminary observations that $\fA_{}'$ is as required.
\end{proof}

\subsection{Decidability and complexity}

From Lemma \ref{twcanona} we get immediately the following
theorem. 
\begin{tw}\label{tw-4S} An $\fodeg$-sentence $\Psi$ is satisfiable if and only if
there exist a $\sigma$-structure $\fA$ and
$S_0,S_1,S_2,S_3\subseteq A$,  such that:
\begin{enumerate}
\item $|A|\leq 4\cdot {\mathbb M},$ 
\item either $S_1=S_2=S_3=\emptyset$, or $\{S_0,S_1,S_2,S_3\}$ is a partition of $A$ and then
\begin{enumerate}
\item $\langle S_0,S_1\rangle \approx_{\fA}\langle
S_0,S_2\rangle\approx_{\fA}\langle S_0,S_1\rangle,$
\item $\langle S_1,S_2\rangle \approx_{\fA}\langle
S_2,S_3\rangle\approx_{\fA}\langle S_1,S_3\rangle,$
\end{enumerate}
\item for every $a,b\in A$, $tp^{\mathfrak A}
[a,b]\models \psi_0$,
\item $T^\fA$ is transitive in $\mathfrak A$,
\item for every $j=0,1,2,$ for every $a\in S_j$
and for every $\gamma_i\in\Psi$,

\hspace{0,5cm}if ${W}_i^{\fA}(a)\cap S_0=\emptyset$, then
${W}_i^{\fA}(a)\cap S_{j+1}\neq \emptyset,$
\item for every  $a\in A$
and for every $\delta_i\in\Psi$, $a$ has a $\delta_i$-witness in $Cl^{\fA}(a)$.

\end{enumerate}
\end{tw}

\begin{proof}
{\bf (}$\mathbf{\Rightarrow}${\bf )} Assume $\fA'$ is a regular
model of $\Psi$ (given by Lemma \ref{twcanona})
 with partition $P_{A'}=\{S_0, S_1,\ldots\}$ and for every $0< k <\infty$:
\begin{enumerate}
\item $\langle S_{k+1},S_{k+2}\rangle \approx \langle S_k,S_{k+1} \rangle,$ 
\item $\langle S_0,S_{k+1}\rangle \approx \langle S_{0},S_{k}\rangle.$
\end{enumerate}
Define $\fA\stackrel{def}{=}\fA'\obciety (S_0\dot\cup
S_1\dot\cup S_2\dot\cup S_3)$. Note that $|A|\leq 4\cdot{\mathbb M}.$ 

\noindent{\bf (}$\mathbf{\Leftarrow}${\bf )}
Define a structure
$\fA'$ such that  $A'\stackrel{def}{=}S_0\;\dot\cup\;
S_1\;\dot\cup\; S_2\;\dot\cup\;
S_3\;{\dot\cup}\;{\dot{\bigcup}}_{j=4}^{\infty}\;S_j$ and, for every $0<j< k <\infty$:
\begin{itemize}
\item $\langle S_j,S_{k}\rangle \approx_{\fA'} \langle S_1,S_{2} \rangle$ and 
\item $\langle S_0,S_{j}\rangle \approx_{\fA'} \langle S_{0},S_{1}\rangle.$
\end{itemize}
Obviously,  $\fA'\models\Psi.$
\end{proof}

\begin{corollary}\label{col-compl}
SAT($\fodeg$) $\in$ {\TwoNExpTime}.
\end{corollary}

\begin{proof} To check whether a given \fodeg{}-sentence  is satisfiable
we take its normal form $\Psi$ and follow Theorem \ref{tw-4S} to obtain a nondeterministic
double exponential time procedure, as described below.
\begin{enumerate}
\item {\tt Guess} a $\sigma$-structure $\fA$ of cardinality
$|A|\leq4\cdot {\mathbb M}$, 

\noindent{\tt guess} and partition
$P_A=\{S_0$,$S_1,S_2,S_3\}$ and

\item {\tt If not}:
\begin{enumerate}
\item $\langle S_1,S_0\rangle \approx\langle
S_2,S_0\rangle\approx\langle S_3,S_0\rangle$ and
\item $\langle S_2,S_1\rangle \approx\langle
S_3,S_2\rangle\approx\langle S_3,S_1\rangle$
\end{enumerate}
{\tt then} {\bf reject};
\item {\tt For every} $a,b\in A,$ {\tt if} $tp^{\mathfrak A}
[a,b]\not\models \psi_0$ {\tt then} {\bf reject};
\item {\tt For every} $a,b,c\in A,$ {\tt if not} ($T^{\mathfrak A}
[a,b]\wedge T^{\mathfrak A} [b,c]\Rightarrow T^{\mathfrak A}
[a,c])$ {\tt then} {\bf reject};
\item {\tt For every} $j=0,1,2,$ {\tt for every} $a\in S_j,$
{\tt for every} $\gamma_i\in\Psi$ such that ${W}_i^{\fA}(a)\cap
S_0=\emptyset$ {\tt if} ${W}_i^{\fA}(a)\cap S_{j+1}= \emptyset$
{\tt then} {\bf reject};
\item {\tt For every} $a\in A,$
{\tt for every} $\delta_i\in\Psi$ {\tt if} $a$ has no $\delta_i$-witness in $Cl^{\fA}(a)$
{\tt then} {\bf reject}; 
\item[]\hspace{-0.85cm} {\bf Accept};
\end{enumerate}
\end{proof}

\section{Discussion}\label{sec:nontr-wit}
The small clique property for \fodt{} implies in particular that in order to extend the decidability result from the fragment with transitive witnesses, $\fodeg$, to full \fodt{} it suffices to consider the situation when the transitive relation is required to be a partial order. Namely, one can reduce the (finite) satisfiability problem for \fodt{} to the (finite) satisfiability problem for \fod{} with one partial order encoding cliques by single elements satisfying some new unary predicates and connection types between cliques by pairs of elements satisfying new binary predicates. This reduction depends of the bound on the size of the cliques and in our case is exponential (see Lemma~5.6~in \cite{P-H18} for a detailed proof). 

We explain below  that the technique from this paper does not generalise to the fragment with free witnesses, giving an example of a satisfiable  \fodeng-formula $\Phi$ that is an axiom of infinity and does not have narrow models in the sense of Definition~\ref{def-narrow}.

It is perhaps worth noting first that the presumably simplest infinity axiom  $\forall x \exists y\; (x<y) \wedge \forall x\; \neg (x<x)$ is not in \fodeng (where the existential quantifier is applied to a subformula not allowed in this fragment).

 The formula $\Phi$ is written  over a signature $\sigma=\sigma_0 \cup \{<\}$, where $\sigma_0$ consists of unary symbols only and $<$ is a
binary predicate interpreted as a \emph{partial order}. 
We will use the abbreviation $x\sim y$ for the formula $x\neq y \wedge \neg (x<y) \wedge \neg (y<x)$ and say \emph {$x$ and $y$ are incomparable}. If $x<y \vee y<x$ then we say \emph {$x$ and $y$ are comparable}.

Let $I=\{0,1,2,3,4\}$, $\sigma_0 = \{A_i: i\in I\}$ and 
let $\varPhi$ be a conjunction of the following sentences.%
\begin{eqnarray}\label{ex1}\setcounter{equation}{1}
&&\forall x \dot{\bigvee}_{i\in I} A_i\,x\label{f-1} \\
&\bigwedge_{i\in I}& \forall x \forall y(A_ix\wedge A_iy\wedge x\neq y)]\rightarrow (x>y \vee y>x)\label{f0}\\
&\bigwedge_{i=0}^2\bigwedge_{j=i+2}^{j=i+3}&\forall x,y\, (A_ix\wedge A_{j}y)\rightarrow (x>y \vee y>x)\label{f2}\\
&\bigwedge_{i\in I}&\forall x A_ix\, \rightarrow [\exists y \,(A_{i\+ 1}y \wedge x\sim y)\wedge \exists y \,(A_{i\m 1}y \wedge x\sim y)]\label{f3}
\end{eqnarray}
In (\ref{f-1}) $\dot{\vee}$ denotes  {\em exclusive or}. In (\ref{f3}) and below, addition and subtraction in subscripts of $A$'s is always understood modulo 5.

Assume $\fA\models \varPhi$.  The realizations of the respective predicate letters in $\fA$ fulfill the following conditions:
\begin{itemize}
	\item by (\ref{f-1})  the sets $A_0, A_1, A_2, A_3$ and $A_4$  constitute a partition of $A$,
	\item by (\ref{f0}), if $a_i, b_i\in A_i$ $(i\in I)$ and $a_i\neq b_i$ then $a_i$ and $b_{i}$ are comparable,
	\item by (\ref{f2}) if $a_0,\ldots, a_4\in A$, where $a_i\in A_i$ ($i\in I$), then 
	each of the pairs $(a_0, a_{2})$, $(a_0, a_{3})$, $(a_1, a_{3})$,  $(a_1, a_{4})$,  $(a_2, a_{4})$ consists of comparable elements,
	\item by (\ref{f3}), if $a\in A_i$ then $a$ has a witness $b\in A_{i\+1}$ such that $a\sim b$ and $a$ has a  witness $e\in A_{i\m 1}$ such that $e\sim a$, for every $i$ $(i\in I)$. 
\end{itemize}
When $a$, $b$ and $e$ are as in the last item above, then we say that $b$ is a {\em right-witness of $a$} and $e$ is a {\em left-witness of $a$}. Note that if $b$ is a right-witness of $a$, then $a$ is a left-witness of $b$. 

Define a $\sigma$-structure $\fC$ with the universe $C=(a_l)_{l\in \Z}$ such that for every $l\in \mathbb{Z}$ (see Figure~\ref{fig1}):
\begin{itemize}
	\item   $a_l\in A_{(l\mod 5)}$,
	\item $a_l\sim a_{l+1}$,
	\item for $m\geq l+2$, set $a_l < a_m$.
\end{itemize}
It is clear  that
$\fC\models\varPhi$. 
Now we show that $\varPhi$ has only infinite models.

\begin{figure}[h]
	
	\begin{center}
		\xymatrix@-1pc{
			\dots
			&   A_{3}
			& 	A_{4}
			&	A_{0}
			&   A_{1}	
			&   A_{2} 
			&   A_{3} 
			&   A_{4} 
			&   A_0
			&   A_1
			&   \ldots\\
			\ldots \ar@/_0.9pc/[rr]^{}
			&   a_{-2}\ar@{.}[r] \ar@/^0.9pc/[rr]^{}\ar@/^1.4pc/[rrr]^{}
			&   a_{-1}\ar@{.}[r] \ar@/_0.9pc/[rr]^{}\ar@/_1.4pc/[rrr]^{}
			&   a_0\ar@{.}[r]    \ar@/^0.9pc/[rr]^{}\ar@/^1.4pc/[rrr]^{}
			&   a_1\ar@{.}[r]    \ar@/_0.9pc/[rr]^{}\ar@/_1.4pc/[rrr]^{}
			&   a_2\ar@{.}[r]    \ar@/^0.9pc/[rr]^{}\ar@/^1.4pc/[rrr]^{}
			&   a_3\ar@{.}[r]    \ar@/_0.9pc/[rr]^{}\ar@/_1.4pc/[rrr]^{}
			&   a_4\ar@{.}[r]    \ar@/^0.9pc/[rr]^{}\ar@/^1.4pc/[rrr]^{}
			&   a_5\ar@{.}[r]    \ar@/_0.9pc/[rr]^{}
			&   a_6
			&   \ldots
		}
	\end{center}
	\caption{\textsf{\footnotesize The model $\fC$ of  $\varPhi$ is a transitive closure of the directed graph above. For every $l\in \mathbb{Z}$, $a_l\in A_{(l\mod 5)}$. Arrows connect comparable elements, in particular $a_{0}<a_{2}$; dotted lines connect incomparable elements.}} \label{fig1}
\end{figure}

Let $\fA\models\varPhi$. We say that a sequence $C=(a_l)_{l\in \Z}$ of elements of $A$ is a {\em quasi-chain} in $\fA$ if, for every $l\in \Z$, $a_{l+1}$ is a right-witness of $a_l$. By clause \eqref{f3}, every model of $\Psi$ contains a quasi-chain $C$.

\begin{claim}\label{quasi-chain}
	Assume $\fA\models\varPhi$ and	$(a_l)_{l\in \Z}$ is a quasi-chain in $\fA$ with $a_0<a_{2}$.  Then for every $l,m\in\Z$ such that $l+2\leq m$ we have $a_l<a_m$.
\end{claim}

\begin{proof}
	Let $(a_l)_{l\in \Z}$ be a quasi-chain in $\fA$ with $a_0<a_{2}$. Observe that by \eqref{f-1} and \eqref{f3} there is an element $a_i\in (a_l)_{l\in \Z}$ such that $a_i\in A_0$. To simplify notation,  w.l.o.g.~assume $a_0\in A_0$, hence by (\ref{f3}), $a_3\in A_3$. By (\ref{f2}) we get $a_0<a_3$ or $a_3<a_0$. If $a_3<a_0$ then $a_3<a_0<a_2$; a contradiction with $a_2\sim a_3$. So we have that $a_0<a_3$. Now we will show that
	
	(*) for every $l\in \Z$ $a_l<a_{l+2}$ and $a_l<a_{l+3}$
	
	For $l>0$ we proceed by induction. Assume $a_l<a_{l+2}$ and $a_l<a_{l+3}$ for every $l=0,1,\ldots, i-1$. Since $a_{i}\in A_{i}$, $a_{i+2}\in A_{i\+2}$ and $a_{i+3}\in A_{i\+3}$ then  by (\ref{f2}) we get:
	
	$\bullet\,\,\,$  $a_i<a_{i+2}$ or $a_{i+2}<a_i$ and
	
	$\bullet\,\,\,$  $a_i<a_{i+3}$ or $a_{i+3}<a_i$.
	\\If $a_{i+2}<a_i$, by inductive hypothesis, we get $a_{i-1}<a_{i+2}<a_i$; a contradiction with $a_{i-1}\sim a_i$, so we have $a_i<a_{i+2}$. Then, if $a_{i+3}<a_i$ again by inductive hypothesis we get $a_{i+3}<a_{i}<a_{i+2}$; a contradiction with $a_{i+2}\sim a_{i+3}$, so we have $a_i<a_{i+3}$.
	
	In the same way we show (*) for $l<0$.
	
	Now by transitivity of $<$  we obtain that $a_l<a_m$, for $l+2\leq m$ .
\end{proof}
Hence we obtain the following
\begin{corollary}
	The sentence $\varPhi$ is an axiom of infinity.
\end{corollary}

We conclude this example noting that 
by Claim \ref{quasi-chain}, no model of $\varPhi$ is narrow in the sense of Definition~\ref{def-narrow}. This is because no elements $a_k, a_l\in A_0$, $a_k<a_l$, of a quasi-chain $C$ in $\fA$ (cf. $a_0$ and $a_5$ in Figure \ref{fig1})  can have the same right-witness (neither $a_1$ nor $a_6$ can be a candidate for the right witness). Hence the technique presented in Section~\ref{sec:narrow} cannot be applied to show decidability for \fodt{} with free witnesses. So, we have

\begin{corollary}
	\fodeng{} has no narrow model property .
\end{corollary}

There is one more interesting observation about \fodeng: the fragment enjoys the {\em small antichain property}. Namely, every satisfiable \fodeng-formula has a model with finite (bounded) antichains. (An {\em antichain} is a set of pairwise incomparable elements.) The property can be shown using similar ideas as in Lemma~\ref{lem:sc} exploiting the normal form for this fragment. This can be contrasted with the expressive power of \fodt, where we can write a formula that induces an infinite antichain in a model. The formula $\Upsilon$ below illustrates the latter claim. 

Let $\sigma_0=\{P,Q\}$, where $P$ and $Q$ are monadic. Let $\Upsilon$ be the conjunction of the 
following
statements: 
\begin{enumerate}[(a)]
	\item  Elements of $P$ form one infinite chain.
	\item  Elements of $Q$ are incomparable.
	\item  Every element of $P$ has an incomparable element in $Q$.
	\item  Every element of $Q$ is smaller than some element in $P$.
\end{enumerate}

In any model satisfying $\Upsilon$ there is an infinite chain of
elements in $P$ that induces an infinite antichain of elements in
$Q$ (see Figure \ref{infaxiom}).

\begin{figure}[thb]
	\xymatrix@-0.4pc{*!<0.5pc,0pc>{\hspace{2cm}} &
		*!<0pc,0pc>{P} & *!<0pc,0pc>{\ldots} & *+{\circ} \ar@{->}[r]^{}
		\ar@{.}[d] & *+{\circ}\ar@{->}[r]^{}\ar@{.}[d] & *+{\circ}
		\ar@{->}[r]^{} \ar@{.}[d] & *+{\circ} \ar@{->}[r]^{} \ar@{.}[d] &
		*+{\circ} \ar@{->}[r]^{} \ar@{.}[d] &
		*+{\circ}\ar@{}[r] \ar@{.}[d] & 
		*!<1pc,0pc>{\ldots}\\
		*!<0.5pc,0pc>{\hspace{1cm}} & *!<0pc,0pc>{Q} & *!<0pc,0pc>{\ldots}
		& *+{\bullet} \ar@{->}[ru]^{} & *+{\bullet}\ar@{->}[ru]^{} &
		*+{\bullet} \ar@{->}[ru]^{} & *+{\bullet} \ar@{->}[ru]^{} &
		*+{\bullet} \ar@{->}[ru]^{}  &
		*+{\bullet}\ar@{}[r] & 
		*!<1pc,0pc>{\ldots}
	}
	\caption{\textsf{\footnotesize A model for $\Upsilon$ where $<$ is the transitive closure of the edge relation depicted by arrows.}}\label{infaxiom}
\end{figure}

\subsection*{Outlook}

Our paper leaves the following questions open: is the satisfiability problem for \fodt{} with free witnesses decidable? and, is the satisfiability problem for full \fod{} with one transitive relation decidable? 
We believe the answer to both questions is positive; in particular, noting the recent result by Pratt-Hartmann \cite{P-H18} who showed that the finite satisfiability problem for $\fodt$ is decidable in 3-\NExpTime, and in  \TwoNExpTime{} if the transitive predicate $T$ is  interpreted as a partial order. However, the techniques used in \cite{P-H18} do not generalise to the case of infinite structures. 

We also remark that as noted by Kiero\'{n}ski and Michaliszyn \cite{KM12} their technique for  deciding satisfiability  of the two-variable universal fragment of first-order logic with constants and a transitive closure operator of a single binary relation could be extended to a fragment corresponding to \fodt{} with transitive witnesses. However, also this approach leaves a gap in complexity between  \TwoNExpTime{} and \TwoExpTime, and does not generalize to full \fodt. 

We conclude recalling that  the status of the satisfiability problem for \fod{} with two linear orders, to the best of our knowledge, also remains open.

\subsection*{Acknowledgement} 
We are very grateful for valuable comments from the anonymous referees that helped to improve both the presentation and the proof strategy from  the first draft.

\subsection*{Funding}This work is supported by
the Polish National Science Centre grants [2013/09/B/ST6/01535 to W.S.]. 

\bibliographystyle{plain}


\end{document}